\documentclass[12pt]{article}
\pdfoutput=1

\usepackage[margin=1.25in]{geometry}
\usepackage{mathrsfs, mathtools}
\usepackage{latexsym, amssymb, amsmath, amsthm}
\usepackage[round]{natbib}

\usepackage{bm}
\usepackage[title,toc,titletoc]{appendix}
\usepackage{titlesec}
\titleformat{\section}[block]{\large\normalfont\filcenter}{\thesection.}{.5em}{}
\titleformat{\subsection}[block]{\normalfont\bfseries\filright}{\thesubsection.}{.5em}{}

\usepackage{graphicx}
\usepackage{caption,subcaption}
\usepackage{comment}

\usepackage[T1]{fontenc}
\usepackage[utf8]{inputenc}
\usepackage{mathtools, slashed}

\usepackage[T1]{fontenc}
\usepackage[utf8]{inputenc}

\usepackage{setspace}
\usepackage{enumitem}
\usepackage[multiple, flushmargin]{footmisc}
\usepackage[colorlinks=true,linkcolor = blue, citecolor=blue, hyperfootnotes=false, hyperindex,breaklinks]{hyperref}
\usepackage[dvipsnames]{xcolor}
\usepackage{soul}

\usepackage{tikz}
\usetikzlibrary{positioning}

\usepackage{pgfplots}
\pgfplotsset{compat=1.14}
\usetikzlibrary{patterns, decorations.pathreplacing, arrows}

\tikzset{
  every node/.style    = {
    text centered,
    line width = .5,
    anchor = center,
  },
  every label/.style   = {
    fill = white, anchor = mid,
  },
  every path/.style   = {
    > = stealth
  },
  point/.style args   = {(#1)#2}{
    rounded corners,
    fill = white,
    minimum height = 20,
    minimum width = 10,
    label = { [name = #1] above:#2 },
  },
  point a/.style args   = {(#1)#2}{
    rounded corners,
    fill = white,
    minimum height = 10,
    minimum width = 20,
    label = { [name = #1] right:#2 },
  },
  point b/.style args   = {(#1)#2}{
    rounded corners,
    fill = white,
    minimum height = 10,
    minimum width = 75,
    label = { [name = #1] right:#2 },
  },
}

\usepackage{etoolbox}

\newtheorem{theorem}{Theorem}

\newtheorem{lemma}{Lemma}
\newtheorem{proposition}{Proposition}
\newtheorem{claim}{Claim}

\theoremstyle{definition}
\newtheorem{definition}{Definition}

\newtheorem{axiom}{Axiom}

\newtheorem{example}{Example}

\def\S{\mathcal{S}}   
    
 \def\cals{\mathcal{S}}

 \def\Re{\mathbb R}

\newcommand{\norm}[1]{\left\langle#1\right\rangle}
\newcommand{\abs}[1]{\left| #1 \right|}

\DeclareMathOperator*{\argmin}{arg\,min}

\usepackage{cleveref}

\usepackage{lmodern}

\begin{document}
\title{Diversity in Choice as Majorization\thanks{For helpful comments and discussions, we are grateful to Lars Ehlers, Matthew Elliot, Fuhito Kojima, Bobby Pakzad-Hurson, and seminar audiences at University of Bonn, Brown University, Durham Market-Design 
Workshop, EC'25, University of Montreal,  Recent Advances in School Choice 2025, RUD, SEA, SITE-Market Design, and University of Tokyo.}}
\author{Federico Echenique\thanks{Department of Economics, UC Berkeley. Contact: \href{mailto:fede@econ.berkeley.edu}{fede@econ.berkeley.edu}} \hspace*{1.25em} Teddy Mekonnen\thanks{Department of Economics, Brown University. Contact: \href{mailto:mekonnen@brown.edu}{mekonnen@brown.edu}}\hspace*{1.25em} M. Bumin Yenmez\thanks{Department of Economics, WUSTL and Business School, Durham University (UK). Contact: \href{mailto:bumin@wustl.edu }{bumin@wustl.edu}.}}

\maketitle
\begin{abstract}
We propose a framework that uses majorization to model diversity and representativeness in school admissions. We generalize the standard notion of majorization to accommodate arbitrary distributional targets, such as a student body that reflects the population served by the school. Building on this framework, we introduce and axiomatically characterize the $r$-targeting Schur choice rule, which balances diversity and priority in admissions. We show that this rule is optimal: any alternative rule must either leave seats unfilled, reduce diversity, or admit lower-priority students. The rule satisfies path independence (and substitutability), which guarantees desirable outcomes in matching markets.  Our work contributes to the ongoing discourse on market design by providing a new and flexible framework for improving diversity and representation.\\

\end{abstract}

\newpage

\section{Introduction}
Diversity and representation are important admissions criteria for many  organizations. School districts and universities aim for their student bodies to be diverse and representative of broader populations along dimensions such as race, ethnicity, socioeconomic status, gender, and geography. At the same time, they aim to allocate  seats to students with the highest priorities, which may be determined by standardized test performance, proximity of residence, sibling attendance, or other institutional rules. This raises a fundamental question: how should an organization balance its diversity objectives with its priorities?

We address this question in the context of school choice by introducing a framework that applies  the theory of majorization (\cite{hardy1952inequalities}; \cite{marshallolkin}) to model diversity; to say when one student body is more diverse than another. Based on majorization as a notion of comparative diversity, we then propose a set of normative principles, formulated as axioms on choice rules, that a school could decide to either adopt or reject. If adopted, the axioms uniquely determine a flexible choice procedure that reconciles, and trades off, diversity and priority considerations. 

The connection between diversity and majorization is motivated by the latter's characterization through \emph{Robin Hood transfers}. Suppose, for example, the socioeconomic composition of the broader population is uniform across three groups: low-income $(\ell)$, middle-income $(m)$, and high-income $(h)$. Consider a student body  $S$ with more type-$m$ students than type-$\ell$ students. Compared to the broader population, type-$m$ students are over-represented relative to type-$\ell$ students. A Robin Hood transfer reduces the number of type-$m$ students and increases the number of type-$\ell$ students by one, while leaving the number of type-$h$ students unchanged and without reversing the inequality in favor of type-$\ell$. This produces a new student body $S'$ with a more equal composition across $\ell$ and $m$, making $S'$ more representative of the broader population than $S$. In this case, we simply say that $S'$ is more diverse than $S$. At the same time, such a transfer implies that the socioeconomic distribution 
of $S'$ is \emph{majorized} by that of $S$. In this way, majorization provides a natural framework for comparative diversity when the benchmark is an equal representation across types.

The standard notion of majorization has, however, one problem. It assumes that diversity is measured relative to an equal-representation benchmark. This is too restrictive for our purposes. In practice, schools often aim to align the composition of their student body with the demographics of the surrounding community \citep{alves}, which is usually not uniform. Therefore, we generalize standard majorization to define \emph{$r$-targeting diversity}, which flexibly accommodates any type distribution ``$r$'' as the benchmark. As with standard majorization, $r$-targeting diversity is defined through Robin Hood transfers, but tailors when one type is considered to be over-represented relative to another based on the benchmark $r$.

Our framework stands in sharp contrast with current practices in school choice. These guard against the under-representation of students whose types belong to  ``protected classes.'' In practice, under-representation is avoided by reserving seats for students from protected classes. To see the difference, suppose type-$\ell$ represents a protected class. If a school reserves seats for low-income students, it essentially introduces a bias that favors type-$\ell$ students without regard for the representation of other types. The school is effectively stating, \emph{as a matter of principle}, that filling all its seats with only type-$\ell$ students, when there are type-$m$ or type-$h$ applicants, is an acceptable outcome; whereas the symmetric outcome where all seats are filled by only type-$m$ or only type-$h$ students when there are type-$\ell$ applicants is deemed unacceptable. 

Although over-representation of type-$\ell$ students may be unlikely, the practice of reserving seats philosophically accepts such an outcome, and communicates this bias to the broader public. By contrast, under $r$-targeting diversity,  a student body composed entirely of type-$\ell$ students is no more diverse than one composed entirely of type-$m$ or type-$h$ students. In other words, our notion of diversity based on majorization treats over- and under-representation symmetrically.

Building on our comparative diversity framework, we propose two axioms on choice rules that capture how a school may balance its diversity objectives with its priorities. The first axiom, \emph{promotes $r$-targeting diversity}, formalizes the school's objective of admitting a representative student body. It requires that no admitted student can be replaced by a rejected student in a way that would strictly improve the diversity of the admitted student body. 

The second axiom turns to the role of priorities. Consider again the possibility of swapping an admitted student for a rejected one. If such a swap is diversity-neutral (neither improving nor worsening diversity), then the school's decision to admit the first student over the second must rest solely on priority considerations. Our second axiom, \emph{$r$-targeting Schur revealed preference}, requires that the priority ranking inferred from this “revealed-preference exercise” is acyclic. In other words, whenever diversity considerations are silent, the school's choice rule must be consistent with some underlying priority ranking. Crucially, we do not assume an exogenously given priority ranking. Rather, we endogenously recover it from the school's choice rule. This is important because it means that our axioms are not parameterized by an arbitrary priority ranking. Instead, the axioms determine a class of choice functions, and each choice function in our class is consistent with some student priority ranking.

\textbf{We have three main results:} First, we show that an efficiency axiom (\emph{non-wastefulness}), together with the two axioms we have mentioned, uniquely pin down a greedy choice rule that we term the \textit{$r$-targeting Schur choice rule}. More precisely, any choice rule satisfying all three axioms is equivalent to the  $r$-targeting Schur choice rule for some student priority ranking.

Second, when a priority ranking is exogenously given, the associated $r$-targeting Schur choice rule is optimal in the sense that any choice 
rule that does not coincide with it must either be wasteful, yield a less diverse  student body, or admit students of lower priority. 

Finally, we show that the $r$-targeting Schur choice rule satisfies path independence (and substitutability), which guarantees the existence of desirable outcomes in matching markets. This result highlights the advantage of our approach over proportionality rules---for example, admitting an equal number of students from each socioeconomic group---which, while intuitive, fail path independence and thus cannot guarantee stable outcomes.\footnote{See the discussion of substitutability in \cite{echyen12}, pp. 
2682-2683.}

We also compare the $r$-targeting Schur choice rule to choice rules with quotas and reserves, two common approaches for balancing priorities with diversity. We show that these rules do not coincide with our choice rule. The contrast highlights the flexibility of our majorization-based approach. Quotas and reserves impose rigid upper and lower bounds on student types, where these allocations are unresponsive to the actual 
type composition of the applicant pool. By contrast, the $r$-targeting Schur choice rule behaves as if the allocation of type-specific seats were endogenously adjusted to the applicant distribution, doing so in a way that targets the representativeness benchmark $r$.

Finally, while we frame our analysis as a controlled school choice problem, it is not the only salient application of our model and results. Our framework applies to various problems in which distributional objectives play an important role: when a firm hires employees, when a grant-giving foundation funds proposals, when health providers ration vaccines, or when immigration officials resettle refugees.

\subsection{Related Literature}
The literature has mostly addressed questions of diversity by taking a mechanism design approach to school choice. This approach was initiated by \citet{abdulson03}, who brought the ideas and methods of modern mechanism design to the study of school choice and also started the formal study of controlled school choice. Since then, the literature has developed around certain models of incorporating diversity considerations into school choice. The focus is, for the most part, on models of school choice via quotas and reserves. See, for example,  \citet{abd2005}, \citet{koj12}, \citet{hayeyi13}, \citet{westkamp10}, \citet{ehayeyi14}, \citet{dogan16}, \citet{fratro:2017}, and \citet{aytur2020}. 

Within the controlled school choice literature, the paper closest to ours is \citet{echyen12}, who introduced the axiomatic approach to controlled school choice and provided characterizations of choice rules based on quotas and reserves using axioms on efficiency and diversity. Relatedly, \citet{imamura2020} and \citet{abdulgrig2021} also provided axiomatic characterizations of choice rules involving reserves and quotas. Choice rules based on reserves have already been implemented to assign students to schools in some countries such as Brazil \citep{aybo16} and Chile \citep{correa19}.

We differ from the existing literature in one key aspect. While the literature on reserves and quotas focuses on avoiding under-representation of protected types, we instead introduce and adapt the theory of majorization as a framework for diversity that treats both under- and over-representation symmetrically. Moreover, as was discussed in the Introduction, our diversity criterion yields the $r$-targeting Schur choice rule, which is significantly more flexible than what is implied by quotas and reserves because it allows the school's choice to respond to the applicant type distribution. \Cref{sec:reservesandquotas} provides a formal comparison to reserves and quotas.

In more recent work, \citet{hakoyeyo2022} study choice rules based on ordinally concave diversity functions. Ordinal concavity does not represent the majorization order. This distinction is best understood by comparing it with Schur concave functions, which are defined by their monotonicity with respect to majorization. It is easy to find examples that are ordinally concave but not Schur concave, and vice versa.\footnote{For example, the \textit{Shannon diversity index} \citep{shannon1948} is Schur concave but not ordinally concave, whereas the \textit{saturated diversity index} \citep{hakoyeyo2022} is ordinally concave but not Schur concave.} While \citet{hakoyeyo2022} focus on choice rules that promote diversity guided by ordinally concave functions, we do not use diversity functions but instead employ majorization to provide an incomplete and relatively coarse diversity ranking of distributions.  Therefore, our approach is conceptually different from theirs. However, our \autoref{thm:merit} is similar to their result that compares choice rules based on priority. Other priority comparison results have been studied by \cite{sonmez/yenmez:22} and \cite{aygun/turhan:2023}. Some of these results, like ours, rely on a matroid structure, but the reasons for why such a structure arise are very different.

Building on the connection between our framework and matroids, our \autoref{prop:greedy} shows that the $r$-targeting Schur choice rule is equivalent to a greedy rule. \cite{Bonet/Arnosti/Sethuraman:2024} shows that the class of such greedy rules is \emph{explainable}: it is always possible to justify to a rejected student that their rejection was due to the admission of a student that dominates in either the priority or diversity criteria. Their insight thus shows that our choice rule is not only optimal, but also has the additional desirable property of being explainable.

\citet{celebi:2024} studies the rationalizability of choice rules in environments where revealed preference from the choice rules must be compatible with some exogenously given preference over sets of students, such as a preference for diversity. He examines which choice rules used in practice, such as reserves and quotas, are rationalizable in this sense.\footnote{A special case of this appears in Online Appendix D of \citet{echyen12}.} Our choice rule is rationalizable because it satisfies path independence \citep{yang2020}.

As we noted earlier, our framework is relevant beyond the school choice literature. Thus, the choice rule that we introduce may prove useful in matching markets with regional distributional constraints such as the Japanese residency matching \citep{kamakoji-basic}, matching public servants with institutions \citep{thakur18}, refugee resettlement \citep{Andersson2019,refugee2023}, pandemic rationing \citep{rationing2022}, and markets with adjustable quotas \citep{kumano/kurino:2022}. 

Finally, our $r$-targeting choice rule is equivalent to a greedy algorithm that optimizes priority subject to a majorization constraint. While such an optimization problem may resemble the recent applications of majorization to mechanism and information design, as seen in \cite{kleiner/moldovanu/strack:2021}, \cite{nikzad2022}, and \cite{arieli2023}, these papers consider a constrained optimization problem where the objective and the constraint are defined over the same domain (specifically, the space of monotone functions). In contrast, our constrained optimization problem features an objective (priority) that is independent of our constraint (diversity). Therefore, while we rely on similar mathematical tools, our approach is conceptually distinct from theirs.

\section{Model}\label{sec:model}

\subsection{Preliminaries} \label{sec:preliminaries}
We denote the set of nonnegative integers by $\mathbb{Z}_+$ and the set of positive integers by $\mathbb{Z}_{++}$. Similarly, we denote the set of nonnegative and positive reals by $\mathbb R_+$ and $\mathbb R_{++}$, respectively. For any $n\in\mathbb Z_{++}$ and 
$x \in \mathbb R^n$, let  $\norm{x}\coloneqq \sum_{i=1}^n x_i$ denote the sum of the coordinates of $x$, and let $x_{[i]}$ denote the $i^{th}$ largest coordinate of $x$ so that $x_{[1]}\geq \ldots \geq x_{[n]}$. For any $i\in \{1,\ldots,n\}$, let $\chi_i\in \mathbb{Z}_+^n$  denote the unit vector where the $i$-th coordinate is one and all other coordinates are zero. The $(n-1)$-dimensional simplex in $\Re^n$ is written as $\Delta^{n-1}\coloneqq \{x\in \mathbb R^n_+: \norm{x}=1\}$.

For any two vectors $x,y\in \mathbb R^n$, we write $x\geq y$ if $x_i\geq y_i$ for every $i \in \{1,\ldots,n\}$, and we write $x>y$ if $x\geq y$ and  $x_i>y_i$ for some $i\in \{1,\ldots,n\}$. Additionally, we say that  $x$ \emph{majorizes} $y$, denoted as $x\succ y$, if 
\begin{enumerate}
\item $\sum_{i=1}^k x_{[i]} \geq \sum_{i=1}^k y_{[i]}$ for every $k\in \{1,\ldots,n-1\}$, and 
\item $\sum_{i=1}^n x_{[i]} = \sum_{i=1}^n y_{[i]}$.
\end{enumerate}
We say that $x$ \emph{strictly majorizes} $y$, denoted as $x\succ_s y$, if $x \succ y$ but $y\mathrel{\slashed \succ} x$. Finally, for any finite set $X$, let $2^X$ denote the power set, and $\abs{X}$ the cardinality, of $X$.

\subsection{Setup}\label{sec:setup}
We consider an environment with one school and a finite, nonempty set of students $\S$. The school has a capacity $q\in \mathbb Z_{++}$ and a priority ranking $P$ on $\S$. The ranking $P$ is a complete, transitive, and antisymmetric binary relation over students. For example, $P$ could be derived from a composite score for each student that reflects standardized test performance, proximity of the student’s home to the school, whether the student has an older sibling at the school, and similar factors. In our main result, the priority ranking is not imposed exogenously; instead, we characterize the choice rules that are consistent with some priority ranking. 

The school values both the priority ranking and the diversity of its student body. A central element of our model is the classification of students into types, which capture attributes such as disability status, ethnicity, gender, race, socioeconomic status, caste, or veteran status. We focus on mutually exclusive types, with each student belonging to exactly one of $n \in \mathbb{Z}_{++}$ possible types. Although this assumption is restrictive, it encompasses two important cases. First, a school may focus narrowly on a single attribute, such as socioeconomic status. In this case, each student can be uniquely assigned to one of three types: low, middle, or high income. Second, a school may adopt a broader approach that considers the intersection of multiple attributes. For instance, if a school values both socioeconomic status and gender, each student can be uniquely assigned to one of six intersectional types: low-income male, middle-income male, high-income male, low-income female, middle-income female, or high-income female.

We index the $n$ possible types by $i \in \mathcal{N} \coloneqq {1,\ldots,n}$. For any nonempty set of students $S \subseteq \S$, we map $S$ to a vector $\xi(S) \in \mathbb{Z}_{+}^n$, where $\xi_i(S)$ denotes the number of type-$i$ students in $S$. We refer to $\xi(S)$ as the \emph{type distribution} of students in $S$. Throughout the paper, we identify the space $\mathbb{Z}_{+}^n$ of nonnegative integer $n$-dimensional vectors with the set of all \emph{feasible} type distributions that may arise in the school’s choice problem.

\section{Diversity as Majorization}\label{sec:majorization}
We develop a notion of comparative diversity by evaluating two student bodies according to their respective type distributions. Formally, we first define a binary relation on the space of all type distributions $\mathbb{R}_+^n$ and then restrict it to the subset of feasible type distributions $\mathbb{Z}_+^n$.

As a benchmark, consider a maximally diverse student body, defined as one in which each type is equally represented. Formally, let $u \in \Delta^{n-1}$ denote the uniform measure, with $u_i = 1/n$ for all $i \in \mathcal{N}$. Given $\lambda > 0$ students, the maximally diverse distribution is $\lambda u \in \mathbb{R}_+^n$. If $\lambda u \notin \mathbb{Z}_+^n$ (i.e., if $\lambda$ is not divisible by $n$), this type distribution is infeasible. Nevertheless, we treat it as the target distribution that a school with $\lambda$ students seeks to approximate.

Now consider any type distribution $x \in \mathbb{R}^n_+$ such that $x_i - \norm{x}{u}_i > x_j - \norm{x}{u}_j$ for some types $i,j \in \mathcal{N}$. Since $u_i=u_j$, the inequality $x_i - \norm{x}{u}_i > x_j - \norm{x}{u}_j$ just means that $x_i>x_j$. In words, type $i$ is overrepresented relative to type $j$ when the uniform measure $u$ defines the target distribution. Writing the inequality as $x_i-\norm{x}{u}_i>x_j-\norm{x}{u}_j$ will soon be useful.

It is easy to see how to modify $x$ so that it is closer to the target distribution. Consider the distribution $y \in \mathbb{R}_+^n$ obtained by shifting a mass $\delta \geq 0$ from type $i$ to type $j$, while leaving all other types unchanged. That is, $y = x + \delta \chi_j - \delta \chi_i$. Since $\norm{x} = \norm{y}$, the relevant target distribution for both $x$ and $y$ is $\norm{x}{u}$. If $\delta \in (0, x_i - x_j)$, the constructed $y$ is closer than $x$ to the target distribution,\footnote{Formally, $||y-\norm{y}{u}||_2 < ||x-\norm{x}{u}||_2$. Moreover, $y$ is (weakly) closer than $x$ to the target distribution in the $\ell^1$-norm, which in our setting is equivalent to the total-variation norm.} suggesting that $y$ is more diverse than $x$ under the uniform-measure benchmark. By contrast, if $\delta = x_i - x_j$, then $y$ is merely a permutation of $x$ in which the mass of types $i$ and $j$ are swapped, suggesting that $x$ and $y$ are equally diverse.

In short, we obtain $y$ from $x$ by transferring mass from a relatively overrepresented type to an underrepresented type. Such transformations of a distribution are called Robin Hood transfers \citep{arnold2012majorization}. The idea was first suggested by \cite{pigou1912wealth} and formalized by \cite{dalton1920measurement} in the context of income inequality. Moreover, \cite{hardy1952inequalities} show that, for any two distributions $x, y \in \mathbb{R}_+^n$, it is possible to obtain $y$ from $x$ through a sequence of Robin Hood transfers if and only if $x$ majorizes $y$. Thus, under a uniform-measure benchmark, majorization provides a natural binary relation for comparing diversity.

For our purposes, however, a notion of diversity based on the uniform measure is too restrictive. For example, a common consideration for a school is to match the type distribution of its student body with that of a broader population \citep{alves}, which is unlikely to have an equal fraction of each type. We therefore introduce a generalized notion of diversity that allows for an arbitrary benchmark measure $r \in \Delta^{n-1}$, where $r_i$ may represent, for example, the fraction of type-$i$ individuals in the general population residing within the school's geographical area. Given $\lambda > 0$ students, the maximally diverse distribution is then $\lambda r \in \mathbb{R}+^n$. Once again, such a type distribution is infeasible if $\lambda r \notin \mathbb{Z}+^n$, but we may treat it as the target distribution for a school with $\lambda$ students.

\smallskip
\begin{definition}\label{def:rht}
Given a measure ${r}\in\Delta^{n-1}$ and two type distributions $x, y\in\mathbb{R}_+^n$, we say that $y$ is obtained from $x$ via an \emph{$r$-targeting Robin Hood transfer} if there exist types $i,j\in\mathcal{N}$ such that $x_i-\norm{x}{r}_i\geq x_j-\norm{x}{r}_j$, and a constant $0\leq \delta\leq x_i-x_j-\norm{x}({r}_i-{r}_j)$ such that $y=x+\delta \chi_j-\delta\chi_i$.
\end{definition}
\smallskip

\autoref{def:rht} is a straightforward extension of the Robin Hood transfer from the uniform measure to an arbitrary measure $r$: the distribution $y$ is constructed by moving mass from some relatively overrepresented type $i$ to an underrepresented type $j$ so that  $y$ is (weakly) closer than $x$ to the target distribution $\norm{x}r$. We recover the standard definition of a Robin Hood transfer in the special case when $r=u$. 

\smallskip
\begin{definition}\label{def:diversity}
Given a measure $r \in \Delta^{n-1}$ and two type distributions $x, y \in \mathbb{R}_+^n$, we say that $y$ is \emph{more $r$-diverse} than $x$, denoted by $y \mathrel{\unrhd_r} x$, if there exists a finite sequence of distributions $\{z_k\}_{k=0}^K$ with $z_0 = x$, $z_K = y$, such that for each $k \in \{1, \ldots, K\}$, $z_k$ is obtained from $z_{k-1}$ via an $r$-targeting Robin Hood transfer.

Furthermore, we say that $y$ is \emph{strictly more $r$-diverse} than $x$, denoted by $y \mathrel{\rhd_r} x$, if $y \mathrel{\unrhd_r} x$ and $x \mathrel{\slashed\unrhd_r} y$.
\end{definition}
\smallskip

We refer to the notion of comparative diversity in \autoref{def:rht} as \emph{$r$-targeting diversity}. Intuitively, $y \mathrel{\unrhd_r} x$ if $y$ can be reached from $x$ through a series of bilateral transfers across types that move $x$ closer to the target distribution $\norm{x} r$. Note that this notion of comparative diversity applies only when $x$ and $y$ are of equal size, i.e., $\norm{x} = \norm{y}$.

\smallskip
\begin{proposition}\label{prop:majorization}
Fix a measure $r\in\Delta^{n-1}$, and let $T^r:\mathbb{R}^n\to\mathbb{R}^n$ be given by 
\[
T^r(x) \coloneqq x+\norm{x}(u-r).
\]
For any two type distributions $x, y\in\mathbb{R}_+^n$, $y$ is {more $r$-diverse} than  $x$ if and only if $T^{r}(x)$ majorizes $T^r(y)$. 
\end{proposition}
\smallskip

\autoref{prop:majorization} formalizes the equivalence between $r$-targeting diversity and majorization. This equivalence implies that the binary relation $\unrhd_r$ is reflexive and transitive, i.e., a preorder, although it is neither antisymmetric nor complete. The transformation $T^r$ linearly translates distributions so that, for $\lambda > 0$ students, the target distribution under measure $r$ is mapped to the target distribution under the uniform measure, i.e., $T^r(\lambda r) = \lambda u$. Thus, for any distribution $x \in \mathbb{R}_+^n$ with $\norm{x} = \lambda$, we have $T^r(x) \succ T^r(\lambda r)$, making $\lambda r$ the maximally diverse distribution for $\lambda$ students. Moreover, $T^r$ preserves student-body size: for any distribution $x$, $\norm{x} = \norm{T^r(x)}$. Finally, when $r = u$, $T^r$ reduces to the identity function, and we recover the standard equivalence between majorization and Robin Hood transfers with respect to the uniform measure.

The following example illustrates that the choice of $r$ can alter the comparative diversity ordering: distributions that are equivalent or incomparable under $u$ may become strictly ranked when evaluated relative to $r$.

\begin{example}\label{example:1}
Let $n=3$. Consider the following three type distributions: $x=(1,4,1)$, $y=(4,1,1),$ and $z=(3,0,3)$. Notice that $x\succ y$ and $y\succ x$ whereas $x\mathrel{\slashed \succ} z$ and  $z\mathrel{\slashed \succ} x$. Thus, under the uniform measure, $x \mathrel{\unrhd_u} y$ and $y \mathrel{\unrhd_u} x$, while $x$ and $z$ are incomparable to each other in terms of $u$-targeting diversity. If the target measure is instead $r=(1/6, 1/2, 1/3)$, the transformed distributions are $T^r(x)=(2, 3, 1)$, $T^r(y)=(5,0, 1)$, and $T^r(z)=(4,-1,3)$. Notice now that $T^r(z) \succ_s T^r(x)$ and $T^r(y) \succ_s T^r(x)$, whereas $T^r(z) \mathrel{\slashed\succ} T^r(y)$ and $T^r(y) \mathrel{\slashed\succ} T^r(z)$.  Hence, under measure $r$, $x \mathrel{\rhd_r} y$ and $x \mathrel{\rhd_r} z$, but $y$ and $z$ are now incomparable in terms of $r$-targeting diversity.
\end{example}\medskip

For the remainder of the paper, we fix a target measure $r \in \Delta^{n-1}$ and consider the restriction of the associated $r$-targeting diversity preorder to the subset of feasible distributions $\mathbb{Z}_+^n$.

\section{School Choice Rules and Diversity}\label{sec:choicerules}
A school can only admit students from among its applicants, so the type distributions it can achieve for its admitted student body are constrained by the applicants’ type distribution. If the applicant type distribution is given by $x\in\mathbb Z_+^n$, then the set of attainable student-body distributions is given by the \emph{budget set}:
\[
B(x)\coloneqq \big\{y\in \mathbb Z_+^n : y\leq x \text{ and } \norm{y} \leq q \big\}.
\]
Within this budget set, the subset of distributions that are maximal with respect to both the size and the diversity of the student body is given by the \emph{$r$-targeting Schur frontier}:
\[
F_r(x)\coloneqq \big\{y\in B(x) : \nexists z \in B(x) \,\text{ s.t. } z \mathrel{>} y \text{ or } z \mathrel{\rhd_r} y \big\}.
\]
In other words, the $r$-targeting Schur frontier is the subset of feasible distributions beyond which the school can neither admit strictly more students nor choose a strictly more $r$-diverse student body.\footnote{In \Cref{sec:frontier}, we provide several properties of the Schur frontier that we use to prove our main results.}

We now turn to our main question: \textbf{How should a school allocate its limited number of seats in a way that balances diversity and priority considerations?} We address this question using an axiomatic approach.

A \emph{choice rule} is a function $C:2^\S \to 2^\S$ that maps each nonempty set of applicants $S\subseteq \S$ into a set of admitted students $C(S)$ such that $C(S)\subseteq S$ and $|C(S)|\leq q$. We impose the natural restriction on the choice rule that a school can only admit students among its applicants, and that it has to respect its capacity.

Our first axiom says that the choice rule should admit as many students as possible, until the school's capacity constraint binds. In other words, no school seat shall go to waste.

\begin{axiom}
A choice rule $C$ is \textbf{non-wasteful} if, for every $S\subseteq \mathcal{S}$, $\abs{C(S)}=\min\{q,\abs{S}\}$.
\end{axiom}

Our second axiom says that, given the choice rule, the diversity of the admitted student body cannot be ``locally'' improved upon by swapping the admission of one student for another. 

\begin{axiom}\label{ax:bbiaseddiv}
A choice rule $C$ \textbf{promotes $r$-targeting diversity} if, for every $S\subseteq \mathcal{S}$, $s \in C(S)$ and $s' \in S\setminus C(S)$ imply $\xi \Bigl( \bigl(C(S) \setminus \{s\} \bigr) \cup \{s'\} \Bigr)\mathrel{\slashed\rhd_r} \xi(C(S))$. 
\end{axiom}

We use the term ``locally'' because Axiom~\ref{ax:bbiaseddiv} does not restrict the possibility that diversity can be improved by swapping the admission decisions of multiple students. Nevertheless, the following lemma shows that our two axioms together imply that the diversity of the admitted student body cannot be ``globally'' improved.

\smallskip
\begin{lemma}\label{lem:frontier-condition}
A choice rule $C$ is non-wasteful and promotes $r$-targeting diversity if and only 
if, for every $S\subseteq \S$, $\xi(C(S))\in F_r(\xi(S))$. 
\end{lemma}
\smallskip

To understand our third axiom, consider an applicant set $S\subseteq \mathcal{S}$ and two students: an admitted student $s\in C(S)$ and a rejected student $s'\in S\setminus C(S)$ such that 
\[
\xi(C(S))\mathrel{\slashed \rhd_r}\xi\left( \bigl(C(S) \setminus \{s\} \bigr) \cup \{s'\} \right).
\]
In words, the school admits $s$ over $s'$ even though doing so does not improve the diversity of the student body. That is, if we counterfactually swapped the admissions of $s$ and $s'$ while holding all else fixed, the resulting diversity would be either the same or incomparable. In such cases, the school’s decision to admit $s$ rather than $s'$ should be interpreted as evidence that $s$ has higher priority than $s'$. By analogy with revealed preference theory, we may say that $s$ has higher \emph{revealed priority} than $s'$. Our third axiom states that the priority ranking uncovered by this revealed preference exercise should be acyclic.

\begin{axiom}
A choice rule $C$ satisfies the \textbf{$r$-targeting Schur-revealed preference axiom} if there are no sequences of students $\{s_k\}_{k=1}^K$ and sets $\{S_k\}_{k=1}^K$ such that, for every $k\in \{1,\ldots,K\}$, 
\begin{center}
$s_k\in C(S_k)$, $s_{k+1} \in S_k\setminus C(S_k)$, and 
$\xi(C(S_k))\mathrel{\slashed \rhd_r}\xi\Bigl( \bigl(C(S_k) \setminus \{s_k\} \bigr) \cup \{s_{k+1}\} \Bigr)$.\end{center}
where $s_{K+1}=s_1$.    
\end{axiom}

Our three axioms capture normative principles particularly well-suited for school choice. The first axiom reflects an efficiency concern: scarce educational resources should be fully utilized, and no seat should remain empty when there are willing applicants. In practice, schools may also face legal or institutional requirements to fill all available seats. The second and third axioms together provide a framework for reconciling two competing objectives in admissions. The second axiom reflects the institution’s diversity goal: conditional on filling its seats, the admitted class should be as $r$-diverse as possible, in the sense that no simple substitution of one student for another could yield a strictly more diverse student body. While administrators may find it infeasible to evaluate diversity under entirely different admissions scenarios, this axiom offers a tractable and transparent yardstick. The third axiom captures the principle that, absent diversity considerations, admissions decisions should maximize priorities consistently. Taken together, the axioms can be interpreted as selecting an outcome on the diversity-priority Pareto frontier.

Having specified our axioms, we next characterize the class of choice rules that satisfy them. To this end, we define the \emph{$r$-targeting Schur choice rule} for a given priority ranking $P$ over $\S$, denoted by $C^P_r$.\footnote{The reference to Schur in the name of the choice rule alludes to “Schur concavity,” a property of functions that are monotone with respect to the majorization order \citep{marshallolkin}.}

\begin{quote}
	\noindent{}{\bf $r$-targeting Schur Choice Rule}\\
	\noindent{}{\bf Input:} A set of students $S\subseteq \mathcal{S}$. Let $k \coloneq \abs{S}$. 
  Label students in $S$ such that $s_1 \mathrel{P} \ldots \mathrel{P} s_k$. \\
    \noindent{}{\bf Initialization:} Set $S_0 \coloneq \emptyset$.\\
	\noindent{}{\bf Steps $i \in \{1,\ldots ,k\}$:} If $\xi\bigl(S_{i-1}\cup \{s_i\}\bigr)\leq x$ for some $x \in F_r\bigl(\xi(S)\bigr)$, then set $S_i \coloneq S_{i-1}\cup \{s_i\}$. Otherwise, set $S_i \coloneq S_{i-1}$.\\
    \noindent{}{\bf Output:} $S_k=C^P_r(S)$.
\end{quote}

Simply put, an $r$-targeting Schur choice rule admits the students with the highest priority, provided that doing so is consistent with a distribution that is maximally $r$-diverse among the non-wasteful distributions in the budget set. An $r$-targeting Schur rule depends on two parameters: the target $r$ and a priority ranking $P$. Our first main result shows that any choice rule satisfying our three axioms must be an $r$-targeting Schur choice rule for some priority ranking.

\smallskip
\begin{theorem}\label{thm:bias-choice-rule} 
A choice rule $C$ is non-wasteful, promotes $r$-targeting diversity, and satisfies $r$-targeting Schur-revealed preference if and only if $C=C^P_r$ for some priority ranking $P$ over $\mathcal{S}$.
\end{theorem}
\smallskip

\autoref{thm:bias-choice-rule} characterizes the choice rules for which there exists a priority ranking that renders the rule equivalent to an $r$-targeting Schur rule. This is a revealed preference exercise, in the sense that the priority ranking is inferred from the choice rule. Once a priority ranking is fixed---either exogenously or endogenously via \autoref{thm:bias-choice-rule}---it can be used to compare the admitted student body under different choice rules. This comparison notion is analogous to first-order stochastic dominance.

\smallskip
\begin{definition}
Given a priority ranking $P$ and two distinct sets of students $S, S' \subseteq \cals$, $S$ \emph{$P$-priority dominates}  $S'$, 
if $\abs{S}\geq \abs{S'}$ and, for every $k \leq \abs{S'}$, if $s_k$ is the student with the $k$-th highest priority ranking in $S$ and $s'_k$ is the student with the $k$-th highest priority ranking in $S'$, then $s_k \mathrel{P} s'_k$ or $s_k=s'_k$. 
\end{definition}
\smallskip

To rephrase, $S$ $P$-priority dominates $S'$ if $S$ contains at least as many students as $S'$, and when the two sets are compared student-by-student in order of priority, the student in $S$ always has weakly higher priority than the corresponding student in $S'$.

Our second main result, \autoref{thm:merit}, compares the outcomes of the $r$-targeting Schur choice rule with those of any alternative choice rule. Whenever the $r$-targeting Schur rule admits a different set of students than the alternative choice rule, then it either admits more students, a more $r$-diverse set of students, or a priority-dominating set of students.

\smallskip
\begin{theorem}\label{thm:merit}
Fix a priority ranking $P$ over $\mathcal{S}$. Let $C$ be any choice rule such that $C(S) \neq C^P_r(S)$ for some $S \subseteq \mathcal{S}$. Then either
\begin{enumerate}[label={$(\roman*)$}]
\item $|C^P_r(S)|>|C(S)|$, or
\item $\xi(C^P_r(S)) \mathrel{\rhd_r} \xi(C(S))$, or
\item $C^P_r(S)$ $P$-priority dominates $C(S)$.
\end{enumerate}
\end{theorem}

Therefore, once our notion of comparative diversity is adopted, \autoref{thm:merit} establishes that the $r$-targeting Schur choice rule always selects an outcome on the diversity-priority Pareto frontier. In other words, any attempt to improve the outcome in terms of diversity or priority necessarily requires either a trade-off along the other dimension or an inefficient allocation in which some seats remain unfilled.

The proof of \autoref{thm:merit} builds on a structural connection to greedy algorithms in matroid theory. Specifically, we show that the maximally $r$-diverse student bodies that can be admitted from a given applicant pool form the bases of a matroid. This structure ensures 
that greedy methods yield optimal admission policies in the sense of \autoref{thm:merit}. To make this connection concrete, we introduce an alternative greedy algorithm and show that it coincides with the $r$-targeting Schur choice rule.

To this end, we first introduce some useful notation and terminology. A 
\emph{matroid} is a pair $(\mathcal{S}, \mathcal{I})$ where $\mathcal{I}$ is a collection of subsets of $\mathcal{S}$. Each element in $\mathcal{I}$ is called an \emph{independent set}. A maximal independent set, in the sense of set inclusion, is called a \emph{base}. The set of bases of a matroid, say $\mathcal{B}$, is characterized by the following \textit{matroid base axioms}: 
\begin{itemize}[parsep = 0pt, listparindent=1em, itemsep=0pt]
\item[(B1)] $\mathcal{B}$ is nonempty and 
\item[(B2)] for every $S\in \mathcal{B}$, $S' \in \mathcal{B}$, and $s \in S \setminus S'$, there exists $s' \in S' \setminus S$ such that 
$(S \setminus \{s\}) \cup \{s'\} \in \mathcal{B}$ and $(S' \setminus \{s'\}) \cup \{s\} \in \mathcal{B}$.  
\end{itemize}

\smallskip
\begin{lemma}\label{lem:base}
Let $S\subseteq \cals$. Then $\mathcal{B}(S) \coloneq \{S' \subseteq S : \xi(S') \in F_r(\xi(S))\}$ 
satisfies the matroid base axioms.
\end{lemma}
\smallskip

We now define the $r$-targeting \emph{greedy choice rule} for a given priority ranking $P$ over $\mathcal S$, which we denote $G^P_r:2^{\mathcal{S}} \rightarrow 2^{\mathcal{S}}$.

\begin{quote}
	\noindent{}{\bf $r$-targeting Greedy Choice Rule}\\ 
	\noindent{}{\bf Input:} A set of students $S\subseteq \mathcal{S}$. Let $k \coloneq \abs{S}$. 
  Label students in $S$ such that $s_1 \mathrel{P} \ldots \mathrel{P} s_k$. \\
    \noindent{}{\bf Initialization:} Set $S_0 \coloneq \emptyset$.\\
	\noindent{}{\bf Steps $i \in \{1,\ldots ,k\}$:} If $S_{i-1} \cup \{s_i\} \subseteq S'$ for some $S' \in \mathcal{B}(S)$,
then set $S_i \coloneq S_{i-1} \cup \{s_i\}$. Otherwise, set $S_i \coloneq S_{i-1}$.\\
    \noindent{}{\bf Output:} $S_k=G^P_r(S)$.
\end{quote}

The next result says that the two choice rules that we have introduced are one and the same.

\smallskip
\begin{proposition}\label{prop:greedy}
Fix a priority ranking $P$ over $\mathcal{S}$. The $r$-targeting greedy choice rule is the same as the $r$-targeting Schur choice rule, i.e., $G^P_r=C^P_r$.
\end{proposition}
\smallskip

We conclude this section with a simple illustration of the $r$-targeting Schur choice rule for a fixed priority ranking.

\begin{example}
There are $n = 2$ student types: blue ($i = 1$) and red ($i = 2$). The school’s capacity is $q = 3$, and the priority ranking $P$ is given by \[s_1 \mathrel{P} s_2 \mathrel{P} s_3 \mathrel{P} s_4 \mathrel{P} s_5.\] The target distribution is $r = u$. There are five applicants, $S = \{s_1, s_2, s_3, s_4, s_5\}$, with $s_1$, $s_2$, and $s_3$ being blue and $s_4$ and $s_5$ being red. Thus, $\xi(S)=(3,2)$ and
\[
B(\xi(S))=\{y\in \mathbb{Z}_+^2: y_1\leq 3,\, y_2\leq 2,\, \text{ and } \, y_1+y_2\leq 3 \}.
\]
A choice rule that admits students based solely on priority would yield $\{s_1, s_2, s_3\}$, consisting only of blue types. By contrast, the $u$-targeting Schur frontier is
\[
F_u(\xi(S))=\{(2,1), (1,2)\},
\]
which, given the budget set, approximates an equal representation of red and blue students. Consequently,  
\[C^P_u(S)=G^P_u(S)=\{s_1, s_2, s_4\}.\] 
In other words, once $s_1$ and $s_2$ are admitted, the choice rule skips over $s_3$ in favor of $s_4$ to achieve greater diversity, even though $s_3$ has strictly higher priority than $s_4$.
\end{example}

\section{Comparison to Reserves and Quotas}\label{sec:reservesandquotas}
Much of the existing literature, as well as many real-world policies, relies on reserves or quotas to reconcile schools' dual objectives. These policies impose rigid constraints on the choice rule, limiting its flexibility and making it difficult to adapt admissions outcomes to the specific composition of applicants.

Given our notion of comparative diversity, \autoref{thm:merit} already establishes that no other choice rule---including rules based on reserves or quotas---can outperform the $r$-targeting Schur choice rule along both diversity and priority dimensions. A natural question arises: is the $r$-targeting Schur rule equivalent to a choice rule with reserves and quotas?

A choice rule with type-specific quotas fails to satisfy non-wastefulness \citep{echyen12}. Intuitively, under any quota policy, one can construct an applicant set in which the type distribution of applicants causes a quota to bind, leaving seats unfilled despite there being more applicants than the school's capacity. In contrast, the $r$-targeting Schur choice rule is non-wasteful and therefore cannot be equivalent to a choice rule that incorporates quotas.

To answer the question, we argue that the $r$-targeting Schur choice rule is not equivalent to any choice rule with reserves. To see this, suppose there are $n = 3$ student types and consider the uniform target $r = u$. Let the priority ranking $P$ place type-1 students above type-2 students, and type-2 students above type-3 students. The school has a capacity of $q = 5$.

A \emph{reserve policy} is a vector $R \in \mathbb{Z}^3_+$ with $\norm{R} \leq q = 5$, where for each $i \in \{1,2,3\}$, $R_i$ represents the number of seats reserved for type-$i$ students. Following \cite{ehayeyi14} and \cite{echyen12}, the rule first admits $\min\{x_i, R_i\}$ students of each type $i$ based on priority. It then fills the remaining
\[
5-\sum_{i\in\mathcal{N}}\min\{x_i, R_i\}
\] 
\emph{open seats} following the remaining students' priority.

Consider three sets of applicants: $S_1$ with $\xi(S_1)=(5,5,0)$, $S_2$ with $\xi(S_2)=(5,0,5)$, and $S_3$ with $\xi(S_3)=(5,3,2)$. Towards a contradiction, suppose 
that the reserve rule coincides with the $u$-targeting Schur rule in all three cases.

For input $S_1$, the budget set is 
    \[
B(\xi(S_1))=\{y\in \mathbb{Z}_+^3: y_1\leq 5,\, y_2\leq 5,\, y_3=0 \text{ and } \, y_1+y_2\leq 5 \},
\]
and the $u$-targeting Schur frontier is $F_u(\xi(S_1))=\{(3,2,0), (2,3,0)\}$, which approximates an equal representation of types 1 and 2 given the budget set. Since type-1 students have a higher priority than type-2 students, $\xi(C^P_u(S_1))=(3,2,0)$.  

Under the reserve rule, all open seats are allocated to type-$1$ students, yielding a distribution $(5-R_2, R_2, 0)$. Hence, the reserve rule coincides with the Schur choice rule if and only if $R_2 = 2$.

For input $S_2$, the budget set is 
    \[
B(\xi(S_2))=\{y\in \mathbb{Z}_+^3: y_1\leq 5,\, y_2=0,\, y_3\leq 5 \text{ and } \, y_1+y_3\leq 5 \},
\]
and the $u$-targeting Schur frontier is $F_u(\xi(S_2))=\{(3,0,2), (2,0,3)\}$, which approximates an equal representation of types 1 and 3 given the budget set. Since type-1 
students have a higher priority than type-3 students, $\xi(C^P_u(S_2))=(3,0,2)$.  

Under the reserve rule, all open seats are allocated to type-$1$ students, yielding 
$(5-R_3, 0, R_3)$. Thus, the reserve rule coincides with the Schur choice rule if and 
only if $R_3 = 2$.

For input $S_3$, the budget set is 
    \[
B(\xi(S_3))=\{y\in \mathbb{Z}_+^3: y_1\leq 5,\, y_2\leq 3,\, y_3\leq 2 \text{ and } \, y_1+y_2+y_3\leq 5 \},
\]
and the $u$-targeting Schur frontier is $F(\xi(S_3))=\{(2,2,1), (1,2,2), (2,1,2)\}$. Since type-1 students have a higher priority than type-2 students, and since type-2 students have a higher priority than type-3 students, we get $\xi(C^P_u(S_3))=(2,2,1)$.  

Under the reserve rule with $R_2 = R_3 = 2$, the resulting distribution is $(1,2,2)$, which contradicts the Schur choice outcome.

To conclude, no choice rule that uses both reserves and quotas 
can replicate the outcome of the $u$-targeting Schur rule across all 
applicant sets. The reason is the rigidity imposed by 
reserves before observing the type distribution of applicants. While unused reserves can 
revert to open seats, the reserves rule lacks the adaptability of the Schur choice rule, 
which tailors target distributions to the realized applicant pool.

\section{Two-sided Markets}\label{sec:markets}
In this section, we consider two-sided matching markets, such as those between students and 
schools. Schools aim to admit students of the highest quality while also maintaining diverse classes. Let $\mathcal{C}$ denote the set of schools. Each student $s \in \mathcal{S}$ has a preference relation over $\mathcal{C} \cup \{s\}$, where $s$ represents the outside option for student $s$, such as homeschooling or attending a private school. A school 
$\alpha \in \mathcal{C}$ is \textbf{acceptable} to student $s \in \mathcal{S}$ if the 
student prefers $\alpha$ to their outside option.

It is well-known in the market-design literature that desirable matchings between students and schools exist, and that the deferred-acceptance (DA) algorithm from \cite{gale62} can be applied, provided school choice rules satisfy \textit{path independence} and \textit{the law of aggregate demand} \citep{hatmil05}.\footnote{Path independence is equivalent to \textit{substitutability} together with a mild consistency axiom. See the discussion in the proof of \autoref{thm:pi}.} Specifically, when school choice rules satisfy path independence, the DA algorithm produces a \textit{stable} matching between students and schools such that:
\begin{enumerate}[label=$(\roman*)$, parsep = 0pt, listparindent=1em, itemsep=0pt]
    \item Each student is either matched with an acceptable school or remains unmatched.
    \item Each school wants to keep the students assigned to it.
    \item No student-school pair exists where the student strictly prefers the school to their assigned match and the school is willing to accept the student.
\end{enumerate}
Moreover, when school choice rules also satisfy the law of aggregate demand, the DA algorithm is \textit{strategy-proof}, meaning no student can achieve a more preferred outcome by misreporting their preferences.

We first define the student-proposing deferred-acceptance algorithm for a given choice rule profile $(C_{\alpha})_{\alpha \in \mathcal{C}}$, and then introduce path independence and the law of aggregate demand.

\begin{quote}
	\noindent{}{\bf Student-Proposing Deferred-Acceptance Algorithm}
    
	\noindent{}{\bf Input:} A profile of student preferences.
    
	\noindent{}{\bf Step 1:} Each student $s\in \mathcal{S}$ proposes to the most preferred acceptable school, if such a school exists. For each school $\alpha \in \mathcal{C}$, let $S_1^{\alpha}$ be the set of students who propose to the school. School $\alpha$ tentatively accepts students in $C_{\alpha}(S_1^{\alpha})$ and permanently rejects $S_1^{\alpha} \setminus C_{\alpha}(S_1^{\alpha})$. If there are no rejections, then stop. 

    \noindent{}{\bf Step $k \geq 2$:} Each student $s\in \mathcal{S}$ who was rejected at Step $k-1$ proposes to the most preferred acceptable school that hasn't rejected them yet, if such a school exists. For each school $\alpha \in \mathcal{C}$, let $S_k^{\alpha}$ be the union of the set of students who proposed to the school at Step $k$ and the set of students tentatively accepted at Step $k-1$. School $\alpha$ tentatively accepts students in $C_{\alpha}(S_k^{\alpha})$ and permanently rejects $S_k^{\alpha} \setminus C_{\alpha}(S_k^{\alpha})$. If there are no rejections, then stop. 
    
    \noindent{}{\bf Output:} The matching produced by the acceptances at the last step.
\end{quote}

The student-proposing deferred acceptance algorithm terminates since there can only be a finite number of rejections. Furthermore, it works well under the following conditions:\footnote{\cite{plott1973path} introduces path independence in the context of social choice. See \cite{chayen17} for a discussion of path independence in a matching context. The law of aggregate demand was introduced by \cite{hatmil05}. It is also known as size monotonicity \citep{alkan03}.}

\begin{axiom}
Choice rule $C$ satisfies \textbf{path independence} if, for every $S,S' \subseteq \mathcal{S}$,
\[C(S \cup S')=C\left(C(S) \cup S'\right).\]
\end{axiom}

\begin{axiom}
Choice rule $C$ satisfies \textbf{the law of aggregate demand} if, for every $S,S' \subseteq \mathcal{S}$, 
\[S\supseteq S'  \implies   \abs{C(S)} \geq \abs{C(S')}.\]
\end{axiom}
Since non-wastefulness implies the law of aggregate demand, the $r$-targeting Schur choice rule satisfies this property. The following result establishes that the $r$-targeting Schur choice rule also satisfies path independence.

\begin{theorem}\label{thm:pi}
The $r$-targeting Schur choice rule satisfies path independence. 
\end{theorem}

Next, we provide an example in which multiple schools employ Schur choice rules with different target levels. We then run the DA algorithm in this setting.

\begin{example}
There are $n=2$ student types: blue ($i=1$) and red ($i=2$). Let $\mathcal{S}=\{s_1, \ldots, s_7\}$ denote the set of students, where $s_1$--$s_4$ are blue and $s_5$--$s_7$ are red. 
Consequently, the distribution $\xi(\mathcal{S})$ is $(4,3)$.

There are two schools, $\alpha$ and $\beta$, each with a capacity of $q=3$ and the 
same priority ranking $P$ given by $s_1 \mathrel{P} \ldots \mathrel{P} s_7.$ Both schools use the $r$-targeting Schur choice rule but differ in their target measures: $r_\alpha = u$ (the uniform measure) and $r_\beta = (1/4, 3/4)$. The respective choice rules are $C_\alpha = C^P_{r_\alpha}$ and $C_\beta = C^P_{r_\beta}$.  

Student preferences are as follows: even-numbered students prefer $\alpha$ over $\beta$, odd-numbered students prefer $\beta$ over $\alpha$, and all students strictly prefer being matched to a school rather than remaining unmatched.

We now run the DA algorithm in this setting. 
\medskip

\noindent \textbf{Round 1:}  
Even-numbered students apply to $\alpha$, so $S_1^\alpha=\{s_2,s_4,s_6\}$. Since $\abs{S_1^\alpha}=q$ and the Schur rule is non-wasteful, $C_\alpha(S^\alpha_1)=\{s_2,s_4,s_6\}$.  

Odd-numbered students apply to $\beta$, so $S_1^\beta=\{s_1,s_3,s_5,s_7\}$. Here $\xi(S_1^\beta)=(2,2)$ and the Schur frontier is $F_{r_\beta}(S_1^\beta)=\{(1,2)\}$. Thus, $C_{\beta}(S^\beta_1)=\{s_1,s_5,s_7\}$, leaving $s_3$ unmatched.\smallskip  

\noindent \textbf{Round 2:}  
Student $s_3$ applies to $\alpha$, so $S_2^\alpha=\{s_2,s_3,s_4,s_6\}$. Now, $\xi(S_2^\alpha)=(3,1)$ and $F_{r_\alpha}(S_2^\alpha)=\{(2,1)\}$. Thus, the school admits $C_{\alpha}(S^\alpha_2)=\{s_2,s_3,s_6\}$, leaving $s_4$ unmatched.  

No new students apply to $\beta$, so $C_{\beta}(S^\beta_2)=\{s_1,s_5,s_7\}$.\smallskip  

\noindent \textbf{Round 3:}  
No new students apply to $\alpha$, so $C_{\alpha}(S^\alpha_3)=\{s_2,s_3,s_6\}$.  

Student $s_4$ applies to $\beta$, so $S_3^\beta=\{s_1,s_4,s_5,s_7\}$. Now, $\xi(S_3^\beta)=(2,2)$ with $F_{r_\beta}(S_3^\beta)=\{(1,2)\}$. Thus, $\beta$ admits $C_{\beta}(S^\beta_3)=\{s_1,s_5,s_7\}$. Student $s_4$ remains unmatched, and the algorithm terminates.  

The final outcome assigns $\{s_2,s_3,s_6\}$ to $\alpha$, $\{s_1,s_5,s_7\}$ to $\beta$, and leaves $s_4$ unmatched.\medskip

\noindent \textbf{Stability:}  
This outcome is stable: students $s_1$, $s_2$, $s_5$, $s_6$, and $s_7$ are matched to their most preferred schools, so any potential blocking pair must involve $s_3$ or $s_4$.  

Student $s_3$ cannot block with $\beta$. First, $s_3$ cannot displace student $s_1$, who also has a blue type and a higher priority. Second, while $s_3$ has a higher priority than both $s_5$ and $s_7$, $s_3$ also cannot displace them since they both have red type, and school $\beta$'s target distribution $r_\beta$ favors red types. By a similar logic, $s_4$ cannot form a blocking pair with school $\beta$.

Furthermore, student $s_4$ cannot form a blocking pair with school $\alpha$. First, student $s_4$ cannot displace students $s_2$ and $s_3$ from the school; the latter two students are also blue types but have a higher priority.  While $s_4$ has a higher priority than $s_6$, $s_4$ also cannot displace the sole type-red student admitted to $\alpha$ because the school favors having an equal proportion of type-red and type-blue students.

Finally, no matched student wants to be unmatched, and no matched school wants to have an empty seat. Hence, the outcome is stable.
\end{example}

\section{Conclusion}
We have introduced a novel framework for thinking about diversity and representation using 
generalized majorization. This approach leads to a new choice rule, the $r$-targeting Schur 
choice rule, which allows institutions to target any desired population ratio $r$ as the 
ideal distribution. The rule is particularly well-suited to contexts where representation 
and diversity are important, offering advantages over conventional choice rules such as those with reserves or quotas. Majorization provides a flexible and symmetric way to measure deviations from a target distribution, treating under- and over-representation equally.

We have justified the $r$-targeting Schur choice rule in scenarios where the ratio $r$ is given, and the priority ranking is derived from the choice rule using revealed preferences. Additionally, it is easy to show that a similar characterization is achievable when the priority ranking is predetermined. This can be done by replacing the revealed-preference axiom with an axiom that ensures compatibility between the revealed preference and the given priority ranking. Furthermore, leveraging this characterization, we can also extend our findings to characterize the student-proposing deferred-acceptance mechanism in which schools employ Schur choice rules.\footnote{A characterization of the deferred acceptance mechanism, as opposed to a choice rule (the focus of our paper), would involve methods recently introduced by \cite{Dogan/Imamura/Yenmez:2025}.}

In conclusion, our work contributes to the ongoing discourse on market design by providing a robust framework for improving diversity and representation through the $r$-targeting Schur choice rule. Future research could explore the empirical implications of our findings and investigate additional applications across different institutional settings.

\newpage
\appendix
\section{Appendix}

\subsection{On the \texorpdfstring{$r$}{}-targeting Schur frontier}\label{sec:frontier}
We provide several results that characterize and detail the properties of the $r$-targeting Schur frontier. These results are important in the proofs of our main results, and may be of independent interest because the frontier captures the maximally $r$-diverse distributions that a school is able to achieve.

\begin{lemma}\label{lem:frontier-bias-norm}
Let $x \in \mathbb Z_+^n$. For all $y \in F_r(x)$, $\norm{y}=\min\{q, \norm{x}\}$.
\end{lemma}
\begin{proof}[Proof of \autoref{lem:frontier-bias-norm}]
Consider any $y\in F_r(x)$ and, for the sake of a contradiction, suppose there exists $z\in B(x)$ such that  $\norm{z}> \norm{y}$. Equivalently,  $\norm{z}\geq \norm{y}+1$. Then, there exists an $i\in\mathcal{N}$ such that $z_i\geq y_i+1$. Define a new distribution $\tilde y=y+\chi_i$. By construction, $\tilde y_i\leq z_i\leq x_i$ and  $\tilde y_j=y_j\leq x_j$ for all $j\neq i$, so $\tilde y\leq x$. Additionally, $\norm{\tilde y}\leq \norm{z}\leq q$. Hence, $\tilde y\in B(x)$ and $\tilde y>y$, which would contradict the fact that $y\in F_r(x)$. Consequently, $\norm{y}=\max_{z\in B(x)}\norm{z}=\min\{q, \norm{x}\}$.
\end{proof}\medskip

\begin{lemma}\label{lem:frontier-bias-permute}
Let $x \in \mathbb Z_+^n$ and suppose $|F_r(x)|>1$. Then there exists a constant $a\in \mathbb{R}$ such that for any $y\in \mathbb Z_+^n$, the following are equivalent:
\begin{enumerate}[label={$(\roman*)$}]
\item  $y\in F_r(x)$.
\item For each distribution $z\in F_r(x)$ with $z\neq y$, there exists an involutory permutation $\pi:\mathcal{N}\to\mathcal{N}$ such that for all $i\in \mathcal{N}$: \footnote{A permutation $\pi:\mathcal{N} \to \mathcal{N}$ is an \textit{involution} if $\pi(\pi(i))=i$ for all $i\in \mathcal{N}$.}
\begin{enumerate}[label={$\alph*.$}]
\item $T^r_i(y)=T^r_{\pi(i)}(z)$.
\item $T^r_i(y)>T^r_i(z)$ implies  $T^r_i(y)=T^r_i(z)+1=a$.
\item $T^r_i(y)<T^r_i(z)$ implies $T^r_i(y)+1=T^r_i(z)=a$.
\end{enumerate}
\end{enumerate}
\end{lemma}
\begin{proof}[Proof of \autoref{lem:frontier-bias-permute}]
Fix any distribution $x\in\mathbb{Z}_+^n$ such that $|F_r(x)|>1$, and let $y\in F_r(x)$. Pick any  distribution $z\in F_r(x)$ with $y\neq z$. Since $\norm{y}=\norm{z}$ by \autoref{lem:frontier-bias-norm}, there exist types $i,j \in\mathcal{N}$ with $i\neq j$ such that $y_i>z_i$ and $y_j<z_j$. Define the nonempty and disjoint sets $\mathbf A\coloneqq \{i\in\mathcal{N}: y_i>z_i\}$ and $\mathbf B\coloneqq \{i\in\mathcal{N}: y_i<z_i\}$.  

\begin{claim}\label{claim:1}
For all $i\in \mathbf A$ and  $j\in \mathbf B$, $y_i = y_j+1+\lfloor \norm{y}(r_i-r_j)\rfloor$ and  $z_j = z_i+1+\lfloor \norm{y}(r_j-r_i)\rfloor$.\end{claim}

\begin{proof}[Proof of \autoref{claim:1}]
We first show that for all $i\in \mathbf A$ and $j\in \mathbf B$,  $T^r_i(y)-T^r_j(y)\leq 1$. To see why, suppose $T^r_{i^*}(y)-T^r_{j^*}(y)>1$ for some $i^*\in \mathbf A$ and $j^*\in \mathbf B$. Notice that the distribution $y+\chi_{j^*}-\chi_{i^*}\in B(x)$ and $T^r(y+\chi_{j^*}-\chi_{i^*})=T^r(y)+\chi_{j^*}-\chi_{i^*}$. Thus, $T^r(y)\succ_s T^r(y+\chi_{j^*}-\chi_{i^*})$, or equivalently, $y+\chi_{j^*}-\chi_{i^*} \mathrel{\rhd_r} y$. We  would then conclude that $y\notin F_r(x)$, yielding a contradiction. Hence,  for all $i\in \mathbf A$ and $j\in \mathbf B$, $T^r_i(y)-T^r_j(y)\leq 1$, or equivalently, $y_i-y_j\leq 1+\norm{y}(r_i-r_j)$. Since $y_i-y_j\in \mathbb{Z}^n$, we can conclude that $y_i - y_j\leq 1+\lfloor \norm{y}(r_i-r_j)\rfloor$. By a symmetric argument, we have $T^r_j(z)-T^r_i(z)\leq 1$, or equivalently, $z_j - z_i\leq 1+\lfloor \norm{y}(r_j-r_i)\rfloor$.

We next show that for any $i\in \mathbf A$ and $j\in \mathbf B$,  $T^r_i(y)-T^r_j(y)>0$. Again, suppose there exists $i^*\in \mathbf A$ and $j^*\in \mathbf B$ such that  $T^r_{i^*}(y)-T^r_{j^*}(y)\leq 0$. By definition of the sets $\mathbf A$ and $\mathbf B$, we would then have 
\[
T^r_{j^*}(z)>T^r_{j^*}(y)\geq T^r_{i^*}(y)>T^r_{i^*}(z),
\]
which would imply $T^r_{j^*}(z)-T^r_{i^*}(z)\geq 2$, leading to a contradiction of the already proven claim that $T^r_{j^*}(z)-T^r_{i^*}(z)\leq 1$. Thus, for all $i\in \mathbf A$ and $j\in\mathbf B$,  $T^r_i(y)-T^r_j(y)>0$, or equivalently, $y_i-y_j>\norm{y}(r_i-r_j)$. Since $y_i-y_j\in \mathbb{Z}^n$, we can conclude that $y_i-y_j\geq 1+\lfloor \norm{y}(r_i-r_j)\rfloor$. By a symmetric argument, we have $T^r_j(z)-T^r_i(z)>0$, or equivalently, $z_j - z_i\geq 1+\lfloor \norm{y}(r_j-r_i)\rfloor$. This concludes the proof to \autoref{claim:1}. \end{proof}

\begin{claim}\label{claim:2} For all $i\in\mathbf A$ and $j\in\mathbf B$, $\norm{y}(r_i-r_j)\in\mathbb{Z}$.
\end{claim}
\begin{proof}[Proof of \autoref{claim:2}]
From \autoref{claim:1}, for all $i\in\mathbf A$ and  $j\in \mathbf B$,  
\begin{align*}
0&<y_i-z_i\\[6pt]
&=y_j-z_j+2+\lfloor \norm{y}(r_i-r_j)\rfloor+\lfloor \norm{y}(r_j-r_i)\rfloor\\[6pt]
&=\left\{\begin{array}{ccc}
y_j-z_j+2 &\mbox{if} & \norm{y}(r_i-r_j)\in\mathbb{Z}\\
y_j-z_j+1 &\mbox{if} & \norm{y}(r_i-r_j)\notin\mathbb{Z}
\end{array} \right..
\end{align*}

Notice that if $\norm{y}(r_i-r_j)\notin\mathbb{Z}$ for some $i\in\mathbf A$ and $j\in\mathbf B$, then $y_j-z_j<0$ and $y_j-z_j+1>0$, which is not feasible as $y_j-z_j\in \mathbb{Z}$. Hence, for all $i\in \mathbf A$ and $j\in\mathbf B$, $\norm{y}(r_i-r_j)\in\mathbb{Z}$.
\end{proof}
\begin{claim}\label{claim:3} For all $i\in\mathbf A$, $T^r_i(y)=T^r_i(z)+1$.  Similarly,  for all $j\in\mathbf B$,  $T^r_j(z)=T^r_j(y)+1$. Consequently, $|\mathbf A|=|\mathbf B|$. 
\end{claim}
\begin{proof}[Proof of \autoref{claim:3}]
From \autoref{claim:2}, we can conclude that for all $i\in\mathbf A$ and $j\in\mathbf B$, 
\[
0<y_i-z_i=y_j-z_j+2.
\]
Therefore,  we have $0<z_j-y_j<2$ for all $j\in\mathbf B$, which implies that $z_j-y_j=1$ and $y_i-z_i=1$. We can then conclude that $T^r_i(y)=T^r_i(z)+1$ for all $i\in\mathbf A$ and $T^r_j(z)=T^r_j(y)+1$ for all $j\in\mathbf B$.

Additionally, since $\norm{y}=\norm{z}$, we have 
\begin{align*}
0&=\norm{y}-\norm{z}\\[6pt]
&=\sum_{i\notin\mathbf A\cup\mathbf B} \underbrace{(y_i-z_i)}_{=0} +\sum_{i\in\mathbf A} \underbrace{(y_i-z_i)}_{=1} +\sum_{j\in\mathbf B} \underbrace{(y_j-z_j)}_{=-1}\\[6pt]
&=|A|-|B|,
\end{align*}
giving us the desired conclusion.\end{proof}

\begin{claim}\label{claim:4} There exists a constant  $a\in \mathbb{R}$ such that $T^r_i(y)=T^r_j(z)=a$ for all $i\in\mathbf A$ and $j\in\mathbf B$.  \end{claim}

\begin{proof}[Proof of \autoref{claim:4}] We first show that for all  $i\in\mathbf A$ and $j\in\mathbf B$, $T^r_i(y)=T^r_j(z)$. To see why, note that for all $i\in\mathbf A$ and $j\in\mathbf B$, 
\begin{align*}
T^r_i(y)&=T^r_j(y)+1+\lfloor\norm{y}(r_i-r_j)\rfloor-\norm{y}(r_i-r_j)\\[6pt]
&=T^r_j(z)+\lfloor\norm{y}(r_i-r_j)\rfloor-\norm{y}(r_i-r_j)\\[6pt]
&=T^r_j(z),
\end{align*}
where the first equality comes from \autoref{claim:1}, the second equality from \autoref{claim:3}, and the last equality from \autoref{claim:2}.

We then show that for all $i\in\mathbf A$,  $T(y)_i=a$ for some $a\in \mathbb{R}$. Such a statement is vacuously true if $|A|=1$. Let $|A|>1$. Suppose there exist $i, k\in\mathbf A$ with $i\neq k$ such that $T^r_i(y)>T^r_k(y)$, which implies that for any $j\in\mathbf B$, 
\begin{align*}
y_i&\geq y_k+1+\lfloor \norm{y}(r_i-r_k) \rfloor \\[6pt]
&=y_j+2+\lfloor \norm{y}(r_k-r_j)\rfloor+\lfloor \norm{y}(r_i-r_k) \rfloor \\[6pt]
&>y_j+1+\lfloor \norm{y}(r_i-r_j)\rfloor,
\end{align*}
where the equality follows from applying \autoref{claim:1} to $k$ and $j$,  and the last inequality follows from  \autoref{claim:2} and the fact that for any $a\in \mathbb{R}$, $1+\lfloor a\rfloor>a\geq \lfloor a\rfloor$. The last inequality however contradicts \autoref{claim:1} applied to $i$ and $j$. Hence, $T^r_i(y)=T^r_k(y)$ for all $i, k\in\mathbf A$, and a symmetric argument establishes that $T^r_j(z)=T^r_k(z)$ for all $j, k\in\mathbf B$.
\end{proof}
We have thus far shown that there exist a constant $a\in \mathbb{R}$ and a bijective mapping $\hat\pi:\mathbf A\to \mathbf B$ such that 
\begin{enumerate}[label={(\alph*)}]
\item For all $i\notin \mathbf A\cup\mathbf B$, $T^r_i(y)=T^r_i(z)$, and 
\item For all $i\in\mathbf A$,  $T^r_i(y)=T^r_i(z)+1=T^r_{\hat\pi(i)}(z)=T^r_{\hat\pi(i)}(y)+1=a$.
\end{enumerate}
Consequently, $T^r(y)$ and $T^r(z)$ are permutations. Furthermore, since the two properties listed above must hold for any two distributions in $F_r(x)$, the constant $a$ is independent of $y$ and $z$. Finally, we define the mapping $\pi:\mathcal{N}\to \mathcal{N}$ given by 
\[
\pi(i)=\left\{\begin{array}{ccc}
i & \mbox{if} & i\notin\mathbf A\cup\mathbf B\\
\hat\pi(i) & \mbox{if} & i\in\mathbf A\\
\hat\pi^{-1}(i)& \mbox{if} & i\in\mathbf B
\end{array}
 \right. ,
\]
which gives us the statement of \autoref{lem:frontier-bias-permute}. It only remains to show that $\pi$ is an involution. By construction, for $i\notin \mathbf A\cup\mathbf B$, $\pi(\pi(i))=\pi(i)=i$.  For $i\in \mathbf A$, $\pi(i)=\hat\pi(i)\in\mathbf B$ and thus $\pi(\pi(i))=\hat\pi^{-1}(\hat \pi(i))=i$. Similarly,  for $i\in\mathbf B$, $\pi(i)=\hat\pi^{-1}(i)\in\mathbf A$ and thus $\pi(\pi(i))=\hat\pi(\hat \pi^{-1}(i))=i$.  This concludes the proof. 
\end{proof}\medskip

By \autoref{lem:frontier-bias-norm}, any type distribution in the ($r$-targeting) Schur frontier is achieved by admitting as many students as the capacity permits. Furthermore, by \autoref{lem:frontier-bias-permute}, any two distributions in a nonsingleton Schur frontier must be ``permutations" of each other, and if the two distributions differ in their $i^{th}$ coordinate, then they must do so by only one. 

Lemmas~\ref{lem:frontier-bias-norm} and~\ref{lem:frontier-bias-permute} provide substantial structure on the Schur frontier, and allow us to provide a simple characterization that relates majorization to the $\ell^2$-norm. This is striking because majorization---and hence, $r$-targeting diversity---is not a complete binary relation, whereas comparisons of feasible distributions based on their proximity in the Euclidean distance to a target distribution yields a total order over $\mathbb{Z}_+^n$. 

To that end, for any $x \in \mathbb Z_+^n$, define
\begin{align}
\label{eq:minimization}
M_r(x)\coloneqq \argmin_{y\in B(x)} \sum_{i=1}^n \big(y_i-\norm{y}r_i\big)^2 \hspace*{.5em} \text{ s.t. } \hspace*{.1em} \norm{y}=\min\{q, \norm{x}\}.
\end{align} 
Since the constraint set in \eqref{eq:minimization} is nonempty and finite,  $M_r(x)$ is nonempty for any $x\in \mathbb Z^n_+$. From the following proposition, we can then also establish that $F_r(x)$ is nonempty for any type distribution $x\in \mathbb Z^n_+$. 

\begin{lemma}\label{prop:frontier-bias-characterization}
For any $x \in \mathbb Z_+^n$, $F_r(x)= M_r(x)$.
\end{lemma}
\begin{proof}[Proof of \autoref{prop:frontier-bias-characterization}]
Fix a feasible distribution $x\in \mathbb{Z}_+^n$, and notice that we can equivalently define $M_r(x)$ as
\[
M_r(x)\coloneqq \argmin_{y\in B(x)} ||T^r(y)||_2 \hspace*{.5em} \text{ s.t. } \hspace*{.1em} \norm{y}=\min\{q, \norm{x}\}.
\]
Let us first show that  $M_r(x)\subseteq F_r(x)$. Fix some $y\in M_r(x)$. Take any $z\in B(x)$, which by definition of the budget set, satisfies $\norm{z}\leq \min\{q, \norm{x}\}=\norm{y}$. Thus, $z\mathrel{\slashed{>}} y$. 

Suppose $z \mathrel{\rhd_r} y$. By \autoref{prop:majorization}, we have $T^r(y) \succ_s T^r(z)$, which further implies that $\norm{z}=\norm{y}$. Thus, for any continuous and strictly convex function $g:\mathbb{R}\to\mathbb{R}$, we have:\footnote{See Lemma C.1.a of \cite{marshallolkin} for the connections between strictly convex functions, strict Schur-convexity, and strict majorization. } 
\[
\sum_{i=1}^n g(T^r_i(y))>\sum_{i=1}^n g(T^r_i(z)).
\]
However, taking $g(X)=X^2$, the above inequality implies $||T^r(z)||_2<||T^r(y)||_2$, which contradicts the fact that $y\in M_r(x)$. To avoid such a contradiction, we must have that $z\mathrel{\slashed \rhd_r} y$. In other words, for each $y\in M_r(x)$, there exists no distribution $z\in B(x)$ with either $z>y$ or $z \mathrel{\rhd_r} y$. Therefore, $M_r(x)\subseteq F_r(x)$. 

Let us now show the converse: $F_r(x)\subseteq M_r(x)$. We have already established that $M_r(x)$ is nonempty and that $M_r(x)\subseteq F_r(x)$. Pick any $z\in M_r(x)$ and any $y\in F_r(x)$. If $y\neq z$, then from \autoref{lem:frontier-bias-permute}, $T^r(y)$ is a permutation of $T^r(z)$, and hence $y\in M_r(x)$, which establishes $M_r(x)\supseteq F_r(x)$.
\end{proof}\medskip

\begin{lemma}\label{lem:frontier-ranking}
Let $x \in \mathbb Z_+^n$, and let $y\in B(x)$ such that $\norm{y}=\min\{q, \norm{x}\}$. If $y\notin F_r(x)$, then
\begin{enumerate}[label={$(\roman*)$}]
\item There exist $i, j\in \mathcal N$ such that  $y+\chi_i-\chi_j\in B(x)$ and  $y+\chi_i-\chi_j \mathrel{\rhd_r} y$, and 
\item For all $z\in F_r(x)$, $z \mathrel{\rhd_r} y$. 
\end{enumerate} 
\end{lemma}

\begin{proof}[Proof of \autoref{lem:frontier-ranking}] Fix a distribution $x\in\mathbb{Z}_+^n$ such that $F_r(x)\neq B(x)$. Pick any $y\in B(x)\backslash F_r(x)$ such that $\norm{y}=\min\{q, \norm{x}\}$.\medskip

\noindent \textbf{(Proof of point $(i)$):} Consider any $y\in B(x)$ such that $\norm{y}=\min\{q, \norm{x}\}$ and $y\notin F_r(x)$. Thus, $y$ satisfies the constraints of the minimization problem \eqref{eq:minimization} but,  by \autoref{prop:frontier-bias-characterization}, $y\notin M_r(x)$. Therefore, there exists a distribution $z\in M_r(x)\subseteq B(x)$ with $||T^r(z)||_2<||T^r(y)||_2$.

Let us define the sets $\mathbf A\coloneqq \{i: y_i>z_i\}$ and $\mathbf B\coloneqq \{i: y_i<z_i\}$. Since $\norm{z}=\norm{y}$ and $z\neq y$, we know $\mathbf A$ and $\mathbf B$ are nonempty sets. Let $y_{i^*}\coloneqq \max\{y_i: i\in\mathbf A\}$ and  $y_{i_*}\coloneqq \min\{y_i: i\in \mathbf B\}$. Notice that for all $i\in \mathbf A$, $T^r_i(z)+1\leq T^r_i(y)\leq T^r_{i^*}(y)$, and hence, 
\begin{align}
\label{eq:3}
T^r_i(y)+T^r_i(z)\leq 2 T^r_{i^*}(y)-1.
\end{align}
Similarly, for $i\in\mathbf B$, $T^r_i(z)\geq T^r_i(y)_i+1\geq T^r_{i_*}(y)+1$, and hence, 
\begin{align}
\label{eq:4}
T^r_i(y)+T^r_i(z)\geq 2 T^r_{i_*}(y)+1.
\end{align}

As a result, we can now express $||T^r(y)||_2-||T^r(z)||_2>0$ as 
\begin{align*}
0<&\sum_{i=1}^n (T^r_i(y))^2-(T^r_i(z))^2\\[6pt]
=&\sum_{i\in \mathbf A\cup\mathbf B} (T^r_i(y))^2-(T^r_i(z))^2\\[6pt]
=&\sum_{i\in\mathbf A}  \underbrace{\Big(T^r_i(y)-T^r_i(z)\Big)}_{>0} \Big(T^r_i(y)+T^r_i(z)\Big)+\sum_{i\in\mathbf B}  \underbrace{\Big(T^r_i(y)-T^r_i(z)\Big)}_{<0} \Big(T^r_i(y)+T^r_i(z)\Big)\\[6pt]
\leq & \Big(2 T^r_{i^*}(y)-1\Big)\sum_{i\in\mathbf A} \Big(T^r_i(y)-T^r_i(z)\Big)+\Big(2 T^r_{i_*}(y)+1\Big)\sum_{i\in\mathbf B} \Big(T^r_i(y)-T^r_i(z)\Big)\\[6pt]
=&2 \Big(T^r_{i^*}(y)- T^r_{i_*}(y)-1\Big)\sum_{i\in\mathbf A} \Big(T^r_i(y)-T^r_i(z)\Big),
\end{align*}
where the first equality follows from the fact that $T^r_i(y)=T^r_i(z)$ for all $i\notin\mathbf A\cup\mathbf B$, the second inequality follows from \eqref{eq:3} and \eqref{eq:4}, and the last equality follows from the fact that $\sum_{i\in\mathbf A} T^r_i(y)-T^r_i(z)=\sum_{i\in\mathbf B}T^r_i(z)-T^r_i(y)$ because $\norm{T^r(z)}=\norm{T^r(y)}$. 

The above display equation yields $T^r_{i^*}(y)- T^r_{i_*}(y)>1$, 
 which implies that $T^r(y) \succ_s T^r(y)+\chi_{i_*}-\chi_{i^*}$, or equivalently, $y+\chi_{i_*}-\chi_{i^*} \mathrel{\rhd_r} y$. Furthermore, $y+\chi_{i_*}-\chi_{i^*}\in B(x)$ because we are taking a seat away from a student of type $i^*$ (which is possible because $y_{i^*}>z_{i^*}\geq 0$ as $i^*\in\mathbf A$) and giving a seat to a student of type $i_*$  (which is possible because $y_{i_*}<z_{i_*}\leq x_{i_*}$ as $i_*\in\mathbf B$). This concludes the proof of \autoref{lem:frontier-ranking}-$(i)$. \medskip

\noindent \textbf{(Proof of point $(ii)$):} Consider any $z\in F_r(x)$, which implies $\norm{z}=\min\{q, \norm{x}\}=\norm{y}$. Since $z\in F_r(x)$, we have that $y\mathrel{\slashed\rhd_r} z$. Additionally, since $y\notin F_r(x)$, we have that $y$ and $z$ are not equally $r$-diverse, i.e., it is not the case that $y\unrhd_r z \unrhd_r y$. Thus, there are only two mutually exclusive possibilities remaining: Either $y$ and $z$ are not comparable in terms of $r$-targeting diversity, or $z \mathrel{\rhd_r} y$. We show that the first case is not possible, and hence, the only possibility is that $z \mathrel{\rhd_r} y$. 

For the sake of a contradiction, suppose $y$ and $z$ are not comparable in terms of $r$-targeting diversity, i.e., $T^r(y)$ and $T^r(z)$ cannot be ordered by majorization. Let $y_0\equiv y$. Since $y_0\notin F_r(x)$, we  know from  \autoref{lem:frontier-ranking}-$(i)$ that there exists $i_i, j_1\in \mathcal N$ such that $y_1\coloneqq y_0+\chi_{i_1}-\chi_{j_1}\in B(x)$ and $y_1 \mathrel{\rhd_r} y_0$. Notice that $z\mathrel{\slashed \unrhd_r} y_1$; Otherwise, $z \mathrel{\rhd_r} y_0$ by transitivity, which goes against the initial assumption that $y_0$ and $z$ are incomparable. Furthermore, $y_1\mathrel{\slashed\unrhd_r} z$;  Otherwise, since $z$ cannot be strictly less $r$-diverse, we would have that $T^r(y_1)$  and $T^r(z)$ are permutations and thus, $z \mathrel{\rhd_r} y_0$ by transitivity. However, this once again would go against the initial assumption that $y_0$ and $z$ are incomparable. Finally, $y_1\notin F_r(x)$; Otherwise, $T^r(y_1)$ and $T^r(z)$ would be permutations by  \autoref{lem:frontier-bias-permute}, which once again would go against the initial assumption that $y_0$ and $z$ are incomparable., Therefore, we conclude that there exists a distribution $y_1\in B(x)$ such that $y_1 \mathrel{\rhd_r} y_0$,  $y_1$ and $z$ are incomparable, and $y_1\notin F_r(x)$. 

Now suppose we have established that for some $k>1$ that there exists a distribution $y_k\in B(x)$ such that $y_{k} \mathrel{\rhd_r} y_{k-1}$,  $y_k$ and $z$ are incomparable, and $y_k\notin F_r(x)$. The above argument establishes that there exists a distribution $y_{k+1}\in B(x)$ such that  $y_{k+1} \mathrel{\rhd_r} y_{k}$,  $y_{k+1}$ and $z$ are incomparable, and $y_{k+1}\notin F_r(x)$. Hence, there would exist an infinite sequence of distinct distributions $\{y_k\}_{k=0}^\infty$ with each $y_k\in B(x)$ and $y_{k+1} \mathrel{\rhd_r} y_k$. However, it is not possible to have such an infinite sequence of distinct distributions that are ranked by $r$-targeting diversity when there are only a finite number of feasible distributions in $B(x)$. Thus, to avoid this contradiction, we must have $z \mathrel{\rhd_r} y$. 
\end{proof}\medskip

\begin{lemma}\label{lem:comparativestat}
    Let $x\in\mathbb{Z}^n_+$ such that $\norm{x}\geq q$. Fix a type $i\in\mathcal{N}$.
    \begin{enumerate}[label=$(\roman*)$]
        \item If $F_r(x)\cap F_r(x+\chi_i)\neq \emptyset$, then $F_r(x)\subseteq F_r(x+\chi_i)$. 
        \item If $F_r(x)\cap F_r(x+\chi_i)=\emptyset$, then for each $y\in F_r(x)$ and  each $z\in F_r(x+\chi_i)$:
        \begin{enumerate}[label=$\alph*.$]
        \item $y_i=x_i$ and $z_i=x_i+1$, and 
        \item For each  $j\in \mathcal{N}$ with $y_j>z_j$, we have that $y+\chi_i-\chi_j\in F_r(x+\chi_i)$ and $z-\chi_i+\chi_j\in F_r(x)$.
        \end{enumerate}
    \end{enumerate}
    \end{lemma}
\begin{proof}[Proof of \autoref{lem:comparativestat}]
 Fix a distribution $x\in\mathbb{Z}^n_+$ with $\norm{x}\geq q$ and a type $i^*\in\mathcal{N}$.\medskip

\noindent\textbf{Case 1:} $F_r(x)\cap F_r(x+\chi_{i^*})\neq \emptyset$.

Consider a distribution $z\in F_r(x)\cap F_r(x+\chi_{i^*})$ and suppose, for the sake of contradiction, that there exists a distribution $y\in F_r(x)\backslash F_r(x+\chi_{i^*})$. By definition, $y\in B(x)\subseteq B(x+\chi_{i^*})$ and $\norm{y}=q=\min\{q,\norm{x+\chi_{i^*}}\}$. Since $z\in F_r(x+\chi_{i^*})$ and $y\notin F_r(x+\chi_{i^*})$, \autoref{lem:frontier-ranking}-$(ii)$ implies that $z \mathrel{\rhd_r} y$. However, because $z\in F_r(x)\subseteq B(x)$, we cannot have both $z \mathrel{\rhd_r} y$ and $y\in F_r(x)$. Thus, to avoid such a contradiction, we must have $ F_r(x)\subseteq F_r(x+\chi_{i^*})$.\bigskip

 \noindent\textbf{Case 2:} $F_r(x)\cap F_r(x+\chi_{i^*})= \emptyset$.

Consider a distribution $y\in F_r(x)$ and $z\in F_r(x+\chi_{i^*})$. Clearly, 
$\norm{y}=\norm{z}=q$. Since $y\in B(x)\subseteq B(x+\chi_{i^*})$, 
$\norm{y}=\min\{q, \norm{x+\chi_{i^*}}\}$, and $y\notin F_r(x+\chi_{i^*})$, 
\autoref{lem:frontier-ranking}-$(ii)$ implies that 
$z \mathrel{\rhd_r} y$. Since $y\in F_r(x)$, this is only possible if $z\notin B(x)$, 
which in turn implies that $z_{i^*}=x_{i^*}+1$.

Moreover, from \autoref{lem:frontier-ranking}-$(i)$, $y\notin F_r(x+ \chi_{i^*})$ 
implies that there exist types $k,l\in \mathcal{N}$ such that $y+\chi_k-\chi_l\in B(x+\chi_{i^*})$ and $y+\chi_k-\chi_l \mathrel{\rhd_r} y$. However, if $y_{i^*}<x_{i^*}$, then it is also the case that $y+\chi_k-\chi_l\in B(x)$, which would then contradict the fact that $y\in F_r(x)$. Hence, to avoid such a contradiction, we must have $y_{i^*}=x_{i^*}$. 

Let $\mathbf{A}\coloneqq \{i\in \mathcal{N}: y_i<z_i\}$ and $\mathbf{B}\coloneqq \{i\in \mathcal{N}: y_i>z_i\}$. We have already established that $i^*\in\mathbf{A}$. Since $\norm{y}=\norm{z}$, the set $\mathbf{B}$ must be nonempty. Pick any $j^*\in \mathbf{B}$. If $z=y+\chi_{i^*}-\chi_{j^*}$, then we trivially have our desired result.

Instead, consider the case with $z\neq y+\chi_{i^*}-\chi_{j^*}$. In this case, there exist finite sequences $\{i_m\}_{m=1}^M$ and $\{j_m\}_{m=1}^M$ with $i_m\in \mathbf{A}$ and $j_m\in\mathbf{B}$ for all $m\in\{1,\ldots, M\}$ such that 
\[
z=y+\chi_{i^*}-\chi_{j^*}+\underbrace{\sum_{m=1}^M \Big(\chi_{i_m}-\chi_{j_m}\Big)}_{\coloneqq w}.
\]
Notice that  $w\neq 0$ but $\norm{w}=0$. By construction, $i^*\notin \mathbf{B}$ and thus $j_m\neq i^*$ for all $m\in \{1,\ldots, M\}$. Also, notice that because $z_{i^*}=y_{i^*}+1$, we must have $i_m\neq i^*$ for all $m\in \{1,\ldots, M\}$. Thus, $y+w\in B(x)$ and $\norm{y+w}=\norm{y}=q$. By \autoref{prop:frontier-bias-characterization},   $y\in M_r(x)$. Consequently, $||T^r(y)||_2\leq ||T^r(y+w)||_2$, or equivalently
\begin{align*}
    \sum_{i\in\mathcal N} T^r_i(y)^2&\leq  \sum_{i\in\mathcal N} T^r_i(y+w)^2=\sum_{i\in\mathcal N} T^r_i(y)^2+T^r_i(w)^2+2T^r_i(y)T^r_i(w),
\end{align*}
where the equality follows from the linearity of $T^r$. As a result, we have 
\begin{align}
\label{eq:inequality1}
    \sum_{i\in\mathcal N} T^r_i(w)^2+2T^r_i(y)T^r_i(w)\geq 0.
\end{align}

Similarly, by \autoref{prop:frontier-bias-characterization}, $z\in M_r(x+\chi_{i^*})$ and $y+\chi_{i^*}-\chi_{j^*}\in B(x+\chi_{i^*})$ with $\norm{z}=\norm{y+\chi_{i^*}-\chi_{j^*}}=q$. Thus, $||T^r(y+\chi_{i^*}-\chi_{j^*})||_2\geq ||T^r(z)||_2= ||T^r(y+\chi_{i^*}-\chi_{j^*}+w)||_2$, or equivalently
\begin{align*}
    \sum_{i\in\mathcal N} T^r_i(y+\chi_{i^*}-\chi_{j^*})^2
    \geq \sum_{i\in\mathcal N} T^r_i(y+\chi_{i^*}-\chi_{j^*})^2+T^r_i(w)^2+2T^r_i(w)T^r_i(y+\chi_{i^*}-\chi_{j^*}).
\end{align*}
Consequently, we have 
\begin{align}
\label{eq:inequality2}
\sum_{i\in\mathcal N} T^r_i(w)^2+2T^r_i(w)T^r_i(y)+2T^r_i(w)T^r_i(\chi_{i^*}-\chi_{j^*})\leq 0.
\end{align}
The inequalities \eqref{eq:inequality1} and \eqref{eq:inequality2} together imply that 
\begin{align}
\label{eq:inequality3}
\sum_{i\in\mathcal N} T^r_i(w)T^r_i(\chi_{i^*}-\chi_{j^*})= \sum_{i\in\mathcal N} \sum_{m=1}^M T^r_i(\chi_{i_m}-\chi_{j_m}) T^r_i(\chi_{i^*}-\chi_{j^*})\leq 0.
\end{align}
However, notice that for each $i\in\mathcal{N}$ and each $m\in\{1,\ldots, M\}$
\[
T^r_i(\chi_{i_m}-\chi_{j_m}) T^r_i(\chi_{i^*}-\chi_{j^*})=\left\{\begin{array}{ccc}
1&\mbox{if}  &i=i_m=i^* \mbox{ or } i=j_m=j^* \\
0     & &\mbox{otherwise} 
\end{array}\right..
\]
Thus, \eqref{eq:inequality3} must hold with equality, i.e., the left-hand-side must equal zero.\footnote{Notice that this is true if and only if $i_m\neq i^*$ (which recall is already the case) and $j_m\neq j^*$ for all $m\in\{1,\ldots, M\}$.} This in turn implies that both \eqref{eq:inequality1} and \eqref{eq:inequality2} must also hold with equality. In particular, the left-hand side of \eqref{eq:inequality1} is zero, which is equivalent to 
\[
||T^r(y)||_2=||T^r(y+w)||_2.
\]
Since $y\in M_r(x)$ and $\norm{y}=\norm{y+w}=q$, we then have that $y+w\in M_r(x)$ as well. By \autoref{prop:frontier-bias-characterization}, we can conclude that $z-\chi_{i^*}+\chi_{j^*}=y+w\in F_r(x)$, as desired.

Similarly, the left-hand side of \eqref{eq:inequality2} is zero, which is equivalent to 
\[
||T^r(z)||_2=||T^r(y+\chi_{i^*}-\chi_{j^*})||_2.
\]
Since $z\in M_r(x+\chi_{i^*})$ and $\norm{z}=\norm{y+\chi_{i^*}-\chi_{j^*}}=q$, we then have that $y+\chi_{i^*}-\chi_{j^*}\in M_r(x+\chi_{i^*})$ as well. By \autoref{prop:frontier-bias-characterization}, we can conclude that $y+\chi_{i^*}-\chi_{j^*}\in F_r(x+\chi_{i^*})$, as desired.
\end{proof}\medskip

\begin{lemma}\label{lem:counterfactual}
Let $R, R'\subseteq S\subseteq \S$ with $R\neq R'$. Suppose $\xi(R)$ and $\xi(R')$ are in $F_r(\xi(S))$. Then for each student $s\in R\setminus R'$, there exists a student  $s'\in R'\setminus R$ such that $\xi((R\setminus \{s\})\cup \{s'\})\in F_r(\xi(S))$ and $\xi((R'\setminus \{s'\})\cup \{s\})\in F_r(\xi(S))$.
\end{lemma}

\begin{proof}[Proof of \autoref{lem:counterfactual}] Since both $\xi(R)$ and $\xi(R')$ are in $F_r(\xi(S))$, we have $\vert R\vert=\vert R'\vert$. Additionally, since $R\neq R'$,  the sets $R\setminus R'$ and $R'\setminus R$ are both nonempty. Pick any student $s\in R\setminus R'$ and let $i^*\in \mathcal N$ be the student's type. 

We first consider the case in which $\xi_{i^*}(R)=\xi_{i^*}(R')$. In other words, $R$ contains the same \emph{number} of type-$i^*$ students but not the same type-$i^*$ students as $R'$. In this case, there exists a student $s'\in R'\setminus R$ with type $i_{s'}$ such that $i_{s'}=i^*$. The statement of the lemma then follows trivially.

Next consider the case in which  $\xi_{i^*}(R)\neq \xi_{i^*}(R')$. Without loss of generality, suppose $\xi_{i^*}(R)> \xi_{i^*}(R')$. Since $\abs{R}=\abs{R'}$,  there exists a type  $i'\in \mathcal N$ with $i'\neq i^*$ such that $\xi_{i'}(R)<\xi_{i'}(R')$. Thus, there exists a student  $s'\in R'\setminus R$ with type $i_{s'}$ such that $i_{s'}=i'$. Since both $\xi(R)\in F_r(\xi(S))$ and  $\xi(R')\in F_r(\xi(S))$, by \autoref{lem:frontier-bias-permute},  $T^r_{i^*}(\xi(R))=T^r_{i^*}(\xi(R'))+1=T^r_{i'}(\xi(R'))=T^r_{i'}(\xi(R))+1$. 

Construct a new distribution $\xi((R\setminus \{s\})\cup\{s'\})=\xi(R)+\chi_{i'}-\chi_{i^*}$. Then notice that, by construction,  
\begin{align*}
T^r_k\big(\xi((R\setminus \{s\})\cup\{s'\})\big)=\left\{\begin{array}{ccc}
T^r_{i'}(\xi(R)) &\mbox{if} & k=i^*\\
T^r_{i^*}(\xi(R)) &\mbox{if} & k=i'\\
T^r_k(\xi(R)) &\mbox{if} & k\notin\{i^*,i'\}
\end{array}\right..
\end{align*}
In other words, $T^r\big(\xi((R\setminus \{s\})\cup\{s'\})\big)$ is a permutation of $T^r(\xi(R))$, and by \autoref{prop:majorization}, $\xi(R)\unrhd_r \xi((R\setminus \{s\})\cup\{s'\})\unrhd_r \xi(R)$. Moreover, $\xi((R\setminus \{s\})\cup\{s'\})\in B(x)$ because $\norm{\xi((R\setminus \{s\})\cup\{s'\})}=\norm{\xi(R)}\leq \min\{q, \norm{x}\}$, and because 
\[
0\leq \xi_{i^*}(R')= \xi_{i^*}((R\setminus \{s\})\cup\{s'\})< \xi_{i^*}(R)\leq x_{i^*}
\]
and 
\[0\leq \xi_{i'}(R)<\xi_{i'}((R\setminus \{s\})\cup\{s'\})= \xi_{i'}(R')\leq x_{i'}.
\]
We can therefore conclude that $\xi((R\setminus \{s\})\cup\{s'\})\in F_r(\xi(S))$. A symmetric argument shows that $\xi((R'\setminus \{s'\})\cup\{s\})\in F_r(\xi(S))$. 
\end{proof}\medskip

\subsection{\texorpdfstring{Proofs of \Cref{sec:majorization}}{}}
\noindent \begin{proof}[Proof of \autoref{prop:majorization}]
Consider two distributions $z, z'\in\mathbb{R}_+^n$ such that $z'$ is obtained from $z$ via a Robin Hood transfer with respect to $r$. Then there exist types $i,j\in\mathcal{N}$ with $z_i-\norm{z}{r}_i\geq z_j-\norm{z}{r}_j$, and a constant $0\leq \delta\leq z_i-z_j-\norm{z}({r}_i-{r}_j)$ such that 
\[
z'=z+\delta \chi_j-\delta\chi_i.
\]
By definition of $T^r$, we have:
\begin{enumerate}[label={$(\alph*)$}] 
    \item $\norm{z}=\norm{T^r(z)}$, 
    \item $z_i-\norm{z}{r}_i\geq z_j-\norm{z}{r}_j\Longleftrightarrow T^r_i(z)-\norm{T^r(z)}u_i\geq T^r_j(x)-\norm{T^r(z)}{u}_j$,
    \item $\delta\leq z_i-z_j-\norm{z}({r}_i-{r}_j)\Longleftrightarrow  \delta\leq T^r_i(z)-T^r_j(z)$, and
    \item $z'=z+\delta \chi_j-\delta\chi_i \Longleftrightarrow T^r(z')=T^r(z)+\delta \chi_j-\delta\chi_i$.
\end{enumerate}
Thus, $z'$ is obtained from $z$ via a Robin Hood transfer with respect to $r$ if and only if $T^r(z')$ is obtained from $T^r(z)$ via a Robin Hood transfer with respect to $u$. 

Consequently, for any two distributions $x,y\in\mathbb{R}_+^n$, $y$ is {more $r$-diverse} than  $x$ if and only $T^r(y)$ is more $u$-diverse than $T^r(x)$. Moreover, from \cite{hardy1952inequalities} (see for example, Theorem B.2 of \cite{marshallolkin}),  $T^r(y)$ is more $u$-diverse than $T^r(x)$ if and only if $T^r(x)$ majorizes $T^r(y)$, which yields the desired conclusion. 
\end{proof}

\subsection{\texorpdfstring{Proofs of \Cref{sec:choicerules}}{}}
\begin{proof}[Proof of \autoref{lem:frontier-condition}]\

\noindent (\emph{If} direction:) Let  $\xi(C(S))\in F_r(\xi(S))$ for all $S\in\S$. From \autoref{lem:frontier-bias-norm}, we have $\norm{\xi(C(S))}=\min\{q, \norm{\xi(S)}\}$. By noting that $\norm{\xi(C(S))}=\vert C(S)\vert$ and $\norm{\xi(S)}=\vert S\vert$, we can conclude that $C$ is a non-wasteful choice rule. Furthermore, by the definition of the $r$-targeting Schur frontier,  any $z\in B(\xi(S))$ satisfies $z\mathrel{\slashed\rhd_r }\xi(C(S))$. Hence, $C$ promotes $r$-targeting diversity. \medskip

 \noindent (\emph{Only if} direction:) Let $C$ be a non-wasteful choice rule that promotes $r$-targeting diversity. Since $C(S)\subseteq S$ and $|C(S)|\leq q$, we have $\xi(C(S))\in B(\xi(S))$. Furthermore, because $C$ is non-wasteful, $\norm{\xi(C(S))}=\min\{q, \norm{\xi(S}\}$. It remains to show that $\xi(C(S))\in F_r(\xi(S))$. For the sake of contradiction, suppose not! Then, by \autoref{lem:frontier-ranking}-$(i)$, there would exist $i, j\in \mathcal N$ such that  $\xi(C(S))+\chi_{i}-\chi_{j}\in B(\xi(S))$ and $\xi(C(S))+\chi_{i}-\chi_{j} \mathrel{\rhd_r} \xi(C(S))$. However, this would contradict the fact that $C$ promotes $r$-targeting diversity. Thus, to avoid the contradiction, we must have $\xi(C(S))\in F_r(\xi(S))$ as desired. 
\end{proof}\medskip

\begin{proof}[Proof of \autoref{thm:bias-choice-rule}]\

\noindent\textit{(If direction)}: Fix some priority ranking $P$ over $\mathcal{S}$. We show that $C^P_r$ satisfies the three axioms. First, by construction, $\xi(C^P_r(S))\in F_r(\xi(S))$ for all $S\subseteq \mathcal{S}$, which by \autoref{lem:frontier-condition}, implies that  $C^P_r$ is non-wasteful and promotes $r$-targeting diversity. 

To show that $C^P_r$ also satisfies $r$-targeting Schur-revealed preference, let $S\subseteq \mathcal{S}$, $s \in C^P_r(S)$, $s' \in S \setminus C^P_r(S)$, and 
\[
\xi((C^P_r(S))\mathrel{\slashed\rhd_r}\xi\Bigl( \bigl(C^P_r(S) \setminus \{s\} \bigr) \cup \{s'\} \Bigr).
\]
From \autoref{lem:frontier-ranking}-$(ii)$, we must have  $\xi\Bigl( \bigl(C^P_r(S) \setminus \{s\} \bigr) \cup \{s'\} \Bigr)\in F_r(\xi(S))$. We want to show that in this case, $s \mathrel{P} s'$. 

For the sake of contradiction, suppose $s' \mathrel{P} s$. Then $s'$ is considered before $s$ by $C^P_r$. Let $\overline S \subseteq C^P_r(S)$ be the subset of admitted students who are ranked higher than $s'$ and let $\underline S \subseteq C^P_r(S)$ be the subset of admitted students who are ranked lower than $s'$. Thus,  $(C^P_r(S) \setminus \{s\} \bigr) \cup \{s'\}=\overline S\cup(\underline S \backslash \{s\})\cup \{s'\}$. However, notice that 
\[
\xi(\overline S \cup\{s'\})\leq \xi(\overline S \cup(\underline S \backslash \{s\})\cup \{s'\})\in F_r(\xi(S)),
\]
which implies that when $C^P_r$ considers $s'$, there is a distribution in the Schur frontier that justifies admitting $s'$. This of course contradicts the initial assumption that $s' \in S \setminus C^P_r(S)$. To avoid this contradiction, it must be that $s \mathrel{P} s'$. 

Consequently, there cannot exist sequences $\{s_k\}_{k=1}^K$ and $\{S_k\}_{k=1}^K$ with $s_{K+1}=s_1$ and, for every $k\in \{1,\ldots,K\}$,
\begin{center}
$s_k\in C^P_r(S_k)$, $s_{k+1} \in S_k\setminus C^P_r(S_k)$, and $\xi(C^P_r(S))\mathrel{\slashed \rhd_r}\xi\Bigl( \bigl(C^P_r(S_k) \setminus \{s_k\} \bigr) \cup \{s_{k+1}\} \Bigr)$
\end{center}
because this would generate a cycle 
\[
s_1 \mathrel{P} s_2 \mathrel{P}\ldots \mathrel{P} s_{K}\mathrel{P} s_1
\]
contradicting the fact that $P$ is a preference relation. 

\noindent\textit{(Only-if direction)}:  Suppose that $C$ is a choice rule that is non-wasteful, promotes $r$-targeting diversity, and satisfies $r$-targeting Schur-revealed preference. From \autoref{lem:frontier-condition}, we know that $\xi(C(S))\in F_r(\xi(S))$ for all $S\subseteq \S$. 

Define the following binary relation $P$: for any $s,s'\in \cals$, say $s \mathrel{P} s'$ whenever there exists $S \subseteq \cals$ such that $s\in C(S)$,  $s'\in S\setminus C(S)$, and $\xi(C(S))\mathrel{\slashed\rhd_r}\xi\bigl((C(S) \setminus \{s\}) \cup \{s'\}\bigr)$. The $r$-targeting Schur-revealed preference implies that $P$ does not have any cycles. Therefore, by Szpilrajn's lemma (see the version in \cite{chambers2016revealed}), it can be completed to a preference relation over $\mathcal S$. Let $P^*$ be such a completion, and let $C^{P^*}_r$ be the $r$-targeting Schur choice rule for priority ranking $P^*$. By construction, $\xi(C^{P^*}_r(S))\in F_r(\xi(S))$ for all $S\subseteq \S$. 

To prove the \emph{Only-if} direction, we show that $C(S)=C^{P^*}_r(S)$ for all $S\subseteq \mathcal{S}$. Suppose, to the contrary, that there exists a set $S\subseteq \mathcal S$ such that $C(S)\neq C^{P^*}_r(S)$.  By \autoref{lem:counterfactual}, there exist students $s\in C(S)\setminus C^{P^*}_r(S)$ and $s'\in C^{P^*}_r(S)\setminus C(S)$ such that $\xi((C(S)\setminus \{s\})\cup\{s'\})\in F_r(\xi(S)$. Consequently, $\xi(C(S))\mathrel{\slashed \rhd_r} \xi\bigl((C(S) \setminus \{s\}) \cup \{s'\}\bigr)$, implying that $s \mathrel{ P } s'$ by the definition of the priority ranking $P$. 

Since $P^*$ is a completion, we must also have $s \mathrel{P^*} s'$, which implies that $s$ is considered before $s'$ by the choice rule $C^{P^*}_r$. Furthermore, again by \autoref{lem:counterfactual}, we have that $\xi((C^{P^*}_r(S)\setminus \{s'\})\cup\{s\})\in F_r(\xi(S))$ and, therefore, $s$ should be chosen by $C^{p^*}_r$ when considered, i.e., $s\in C^{P^*}_r(S)$. However, this contradicts the fact that $s\in C(S)\setminus C^{P^*}_r(S)$. To avoid such a contradiction, we must have that $C(S)=C^{P^*}_r(S)$ for all $S\subseteq \S$, which concludes the proof. 
\end{proof}\medskip

\begin{proof}[Proof of Theorem \ref{thm:merit}]
Fix a priority ranking $P$ over $\mathcal{S}$. Let $C$ be a choice rule such that for some $S\subseteq \mathcal{S}$, $C(S)\neq C^{P}_r(S)$. Suppose $\abs{C(S)}\geq \abs{C^P_r(S)}$ and $\xi(C^P_r(S))\mathrel{\slashed\rhd_r}\xi(C(S))$. To prove the theorem, it suffices to show that in this case, $C^P_r(S)$ $P$-priority dominates $C(S)$. 

Consider $\mathcal B(S) = \{S' \subseteq S : \xi(S') \in F_r(\xi(S)) \}$. By construction,  $\xi(C^P_r(S)) \in F_r(\xi(S))$. Additionally, $\xi(C(S)) \in F_r(\xi(S))$ by \autoref{lem:frontier-ranking}-$(ii)$. Therefore, $C^P_r(S)$ and $C(S)$ are in  $\mathcal B(S)$.  By \autoref{lem:base}, $\mathcal B(S)$ is a set of bases of a matroid.  Furthermore, by  \autoref{prop:greedy}, $C^P_r=G^P_r$, where $G^P_r$ is the greedy choice rule defined on the matroid with the set of bases $\mathcal B(S)$. We then get the desired result from the following:
\begin{lemma}[\cite{gale1968}]\label{lem:Gale}
The greedy choice rule $G^P_r(S)$ $P$-priority dominates any independent subset of $S$. 
\end{lemma}

Thus, by Lemma \ref{lem:Gale}, $C^P_r(S)$ $P$-priority dominates any independent subset of $S$. Since $C(S)$ is independent (because $C(S)\in \mathcal{B}(S)$), we get that $C^P_r(S)$ $P$-priority dominates $C(S)$.
\end{proof}\medskip

\begin{proof}[Proof of \autoref{lem:base}]
For $S\subseteq \S$, let $\mathcal B(S) \coloneq \{S' \subseteq S : \xi(S') \in F_r(\xi(S)) \}$. 
By  \autoref{prop:frontier-bias-characterization}, $F_r(\xi(S))$ is nonempty. As a result, $\mathcal B(S)$ is nonempty, which means that 
$\mathcal B(S)$ satisfies matroid base axiom B1.

Let $S'\in \mathcal B(S)$, $S'' \in \mathcal B(S)$, and $s'\in S'\setminus S''$. By definition of $\mathcal B(S)$, $\xi(S'),\xi(S'') \in F_r(\xi(S))$. 
By Lemma \ref{lem:counterfactual}, there exists $s'' \in S'' \setminus S'$ such that 
$\xi((S'\setminus \{s'\}) \cup \{s''\}) \in F_r(\xi(S))$ and $\xi((S'' \setminus \{s''\}) \cup \{s'\})\in F_r(\xi(S))$. 
These imply that $(S'\setminus \{s'\}) \cup \{s''\} \in \mathcal B(S)$ and $(S'' \setminus \{s''\}) \cup \{s'\} \in \mathcal B(S)$. Therefore, $\mathcal B(S)$ satisfies matroid base axiom B2.
\end{proof}\medskip

\begin{proof}[Proof of \autoref{prop:greedy}]
Let $S\subseteq \cals$ be the input set. We show that $G^P_r$ and $C^P_r$ coincide at each step $i$ of the respective algorithms by showing that either the same student $s_i$ is added to the set of chosen students or $s_i$ is skipped. Since both algorithms start with the empty set, we get the result that the two choice rules produce the same outcome for input $S$.

At step $i$ of the $r$-targeting greedy choice rule, $s_i$ is added to $S_{i-1}$ if and only if $S_{i-1} \cup \{s_i\} \subseteq S'$ for some $S' \in \mathcal{B}(S)$. By definition,  $S' \in \mathcal{B}(S)$ if and only if $\xi(S')\in F_r(\xi(S))$. In other words,  $s_i$ is added to $S_{i-1}$  if and only if $\xi(S_{i-1} \cup \{s_i\}) \leq x$ for some $x \in F_r(\xi(S))$, which is step $i$ of the $r$-targeting Schur choice rule.
\end{proof}

\subsection{\texorpdfstring{Proofs of \Cref{sec:markets}}{}}

\begin{proof}[Proof of \autoref{thm:pi}]
We need the following properties of choice rules in our proofs. A choice rule $C$ satisfies \textbf{consistency}, if, for each $S \subseteq \mathcal{S}$ and $s \in \mathcal{S} \setminus S$, $s \notin C(S\cup \{s\}) \; \implies \; C(S)=C(S\cup \{s\})$ 
\citep{chayen17}.\footnote{In the context of matching with contracts, \cite{aygson12a} call this property \textit{the irrelevance of rejected contracts}.} A choice rule $C$ satisfies  \textbf{substitutability}, if, for each $S \subseteq \mathcal{S}$ and $s \in \mathcal{S} \setminus S$, $  C(S\cup \{s\}) \backslash \{s\}\subseteq C(S)$ \citep{kelso82,roth84}. 

\begin{lemma}[\cite{aizmal81}]\label{lem:pi}
A choice rule satisfies path independence if and only if it satisfies consistency and substitutability.
\end{lemma}

Using \autoref{lem:pi}, we prove that $C^P_r$ satisfies path independence by showing that it jointly satisfies consistency and substitutability. Fix $S \subseteq \cals$ and $s \in \cals \setminus S$. Let $\xi^1 \coloneqq \xi(C^P_r(S))$ and $\xi^2 \coloneqq \xi(C^P_r(S\cup\{s\}))$. 
\medskip

\noindent \emph{Proof of consistency:}

Suppose $s \notin C^P_r(S\cup \{s\})$. We show $C^P_r(S)=C^P_r(S \cup \{s\})$. Since  $s \notin C^P_r(S\cup \{s\})$, we get $\xi^2 \in B(\xi(S))$. As $B(\xi(S))\subseteq  B(\xi(S\cup\{s\}))$, and $\xi^2\in  F_r(\xi(S\cup\{s\}))$, there is no type distribution in $B(\xi(S))$ that either has strictly more students or is strictly more $r$-diverse than $\xi^2$. Hence, $\xi^2 \in F_r(\xi(S))$. This has two implications. 

First, $\xi^2 \in F_r(\xi(S))$ implies $C^P_r(S \cup \{s\}) \in \mathcal B(S)\coloneqq \{S' \subseteq S : \xi(S') \in F_r(\xi(S))\}$. By \autoref{prop:greedy} and \autoref{lem:Gale}, $C^P_r(S)$ $P$-priority dominates $C^P_r(S \cup \{s\})$. 

Second, $\xi^2 \in F_r(\xi(S))$ implies $F_r(\xi(S))\cap F_r(\xi(S\cup\{s\}))$ is nonempty. Since $C^P_r$ is non-wasteful and $s \notin C^P_r(S\cup \{s\})$, we have that  $\norm{\xi^1}=\norm{\xi^2}=q\leq \norm{\xi(S)}$. By \autoref{lem:comparativestat}-$(i)$, we have that $F_r(\xi(S))\subseteq F_r(\xi(S\cup\{s\}))$. Thus, $\xi^1\in F_r(\xi(S\cup\{s\}))$, which in turn implies that $C^P_r(S) \in \mathcal B(S\cup\{s\})$. Again, applying \autoref{prop:greedy} and \autoref{lem:Gale}, we have $C^P_r(S\cup\{s\})$ $P$-priority dominates $C^P_r(S)$. 

The results that $C^P_r(S)$ and $C^P_r(S \cup \{s\})$ $P$-priority dominate each other imply that $C^P_r(S)=C^P_r(S\cup\{s\})$. This follows because when two sets $P$-priority dominate each other, they must have the same number of students and, furthermore, because $P$ is strict, they must have the same set of students. Therefore, $C^P_r$ satisfies consistency. \medskip

\noindent \emph{Proof of substitutability:} 

If $s\notin C^P_r(S \cup \{s\})$, then, by consistency, we have $C^P_r(S)=C^P_r(S\cup\{s\})$, so substitutability holds trivially. Likewise, if $|S|\leq q$, then $C^P_r(S)=S$ by non-wastefulness. In this case, substitutability again holds trivially because $C^P_r(S \cup \{s\})\subseteq S\cup\{s\}$. 

For the rest of the proof, we assume that $s\in C^P_r(S \cup \{s\})$ and $\abs{S}>q$. In this case, $\norm{\xi^1}=\norm{\xi^2}=q$. Since $s\in C^P_r(S\cup\{s\})\backslash C^P_r(S)$, the set $C^P_r(S)\backslash C^P_r(S\cup\{s\})$ is nonempty. We proceed by showing that there exists a student $s'\in C^P_r(S)\backslash C^P_r(S\cup\{s\})$ such that 
\[
C^P_r(S)\backslash\{s'\}=C^P_r(S\cup\{s\})\backslash\{s\}.
\]
We consider two cases depending on whether $F_r(\xi(S))\cap F_r(\xi(S\cup\{s\}))$ is empty or not.\smallskip

\noindent \textbf{Case 1: }$F_r(\xi(S))\cap F_r(\xi(S\cup\{s\}))\neq \emptyset$. 

By \autoref{lem:comparativestat}-$(i)$, $F_r(\xi(S))\subseteq F_r(\xi(S\cup\{s\}))$. Thus, $\xi^1\in F_r(\xi(S\cup\{s\}))$, which in turn implies that $C^P_r(S)\in\mathcal{B}(S\cup\{s\})$. As $\mathcal{B}(S\cup\{s\})$ satisfies the matroid base axioms (\autoref{lem:base}), there exists a student $s'\in C^P_r(S)\backslash C^P_r(S\cup\{s\})$ such that $(C^P_r(S)\backslash\{s'\})\cup\{s\}\in \mathcal{B}(S\cup\{s\})$ and $(C^P_r(S\cup\{s\})\backslash\{s\})\cup\{s'\}\in \mathcal{B}(S\cup\{s\})$. 

Given  $(C^P_r(S)\backslash\{s'\})\cup\{s\}\in \mathcal{B}(S\cup\{s\})$, by \autoref{prop:greedy} and \autoref{lem:Gale}, we have $C^P_r(S\cup\{s\})$ $P$-priority dominates $(C^P_r(S)\backslash\{s'\})\cup\{s\}$, or equivalently, $C^P_r(S\cup\{s\})\backslash\{s\}$ $P$-priority dominates $C^P_r(S)\backslash\{s'\}$.

Next, since $(C^P_r(S\cup\{s\})\backslash\{s\})\cup\{s'\}\subseteq S$ and is a base in $\mathcal{B}(S\cup\{s\})$, it must also be a base in $\mathcal{B}(S)$. Thus, by \autoref{prop:greedy} and \autoref{lem:Gale}, $C^P_r(S)$ $P$-priority dominates $(C^P_r(S\cup\{s\})\backslash\{s\})\cup\{s'\}$, or equivalently, $C^P_r(S)\backslash\{s'\}$ $P$-priority dominates $C^P_r(S\cup\{s\})\backslash\{s\}$. 

The results that $C^P_r(S)\backslash\{s'\}$ and $C^P_r(S \cup \{s\})\backslash\{s\}$ $P$-priority dominate each other imply that $C^P_r(S)\backslash\{s'\}=C^P_r(S \cup \{s\})\backslash\{s\}$, and thus, $C^P_r(S \cup \{s\})\backslash\{s\}\subseteq C^P_r(S)$. \smallskip

\noindent \textbf{Case 2: }$F_r(\xi(S))\cap F_r(\xi(S\cup\{s\}))= \emptyset$. 

Let $i^*\in\mathcal{N}$ be the type of student $s$. By \autoref{lem:comparativestat}-$(ii.a)$, $\xi_{i^*}^1=\xi_{i^*}(S)$ and $\xi_{i^*}^2=\xi_{i^*}(S\cup\{s\})=\xi_{i^*}^1+1$. Since $\norm{\xi^1}=\norm{\xi^2}=q$, there exists a type $j^*\in\mathcal{N}\backslash\{i^*\}$ such that $\xi^1_{j^*}>\xi_{j^*}^2$. Consider any type $j^*$ student $s'\in C^P_r(S)\backslash C^P_r(S\cup\{s\})$. By \autoref{lem:comparativestat}-$(ii.b)$, $\xi^1+\chi_{i^*}-\chi_{j^*}\in F_r(\xi(S\cup\{s\}))$ and $\xi^2-\chi_{i^*}+\chi_{j^*}\in F_r(\xi(S))$. In other words, $(C^P_r(S)\backslash\{s'\})\cup\{s\}\in \mathcal{B}(S\cup\{s\})$ and $(C^P_r(S\cup\{s\})\backslash\{s\})\cup\{s'\}\in \mathcal{B}(S)$. An argument similar to the one in Case 1 then establishes that $C^P_r(S)\backslash\{s'\}$ and $C^P_r(S \cup \{s\})\backslash\{s\}$ $P$-priority dominate each other, implying that $C^P_r(S)\backslash\{s'\}=C^P_r(S \cup \{s\})\backslash\{s\}$. Thus, $C^P_r(S \cup \{s\})\backslash\{s\}\subseteq C^P_r(S)$. This concludes the proof that $C^P_r$ satisfies substitutability.
\end{proof}

\bibliographystyle{plainnat}
\nocite{}\bibliography{matching}
\end{document}